\documentclass[11pt,letterpaper]{article}
\usepackage{floatrow}
\usepackage{typearea}
\usepackage{pgfplots}
\usepackage{tikz}
\usetikzlibrary{patterns}
\usepackage{enumitem}

\usepackage{subcaption}
\usepackage{float}
\usepackage{bbm}
\usepackage[ruled,linesnumbered,vlined]{algorithm2e}
\typearea{14}

\usepackage{comment,multicol}

\usepackage{multicol}

\usepackage[ruled,linesnumbered,vlined]{algorithm2e}

\usepackage{multirow}
\usepackage{wrapfig}
\usepackage{color,colortbl}
\usepackage{float}
\usepackage{amsthm}
\usepackage{amsmath}
\usepackage{amssymb}
\usepackage{graphicx}
\usepackage{xspace}
\usepackage{nicefrac}
\usepackage{thm-restate}
\usepackage[colorinlistoftodos,prependcaption,textsize=tiny]{todonotes}
\newcommand{\atodo}[1]{\todo[linecolor=red,backgroundcolor=green!25,bordercolor=red]{\textbf{AF:} #1}}
\newcommand{\atodoin}[1]{\todo[linecolor=red,backgroundcolor=green!25,bordercolor=red,inline]{\textbf{AF:} #1}
}

\newcommand{\mmtodo}[1]{\todo[linecolor=red,backgroundcolor=red!25,bordercolor=red,inline]{\textbf{MM:} #1}}

\definecolor{forestgreen}{rgb}{0.13, 0.55, 0.13}

\SetCommentSty{mycommfont}

\usepackage{dsfont}
\usepackage{thmtools}
\usepackage{thm-restate}
\declaretheorem[name=Lemma]{lemma}

\definecolor{Darkblue}{rgb}{0,0,0.4}
\definecolor{Brown}{cmyk}{0,0.61,1.,0.60}
\definecolor{Purple}{cmyk}{0.45,0.86,0,0}

\usepackage[colorlinks,linkcolor=Darkblue,filecolor=blue,citecolor=blue,urlcolor=Darkblue,pagebackref]{hyperref}
\usepackage[nameinlink]{cleveref}

\SetKwFor{Loop}{repeat}{}{}

\newcommand{\initOneLiners}{%
	\setlength{\itemsep}{0pt}
	\setlength{\parsep }{0pt}
	\setlength{\topsep }{0pt}
}

\newcommand{\initTwoLiners}{%
	\setlength{\itemsep}{1pt}
	\setlength{\parsep }{0pt}
	\setlength{\topsep }{0pt}
}

\newtheorem{theorem}{Theorem}
\newtheorem{corollary}{Corollary}

\newtheorem{definition}{Definition}

\newtheorem{conjecture}{Conjecture}

\numberwithin{equation}{section}

\newcommand{\namedref}[2]{\hyperref[#2]{#1~\ref*{#2}}}

\newcommand{\wt}{\widetilde}

\newcommand{\bal}{{\rm bal}}

\newcommand{\E}{{\mathbb{E}}}
\newcommand{\N}{\mathbb{N}}
\newcommand{\R}{\mathbb{R}}

\newcommand{\poly}{{\rm poly}}
\newcommand{\polylog}{{\rm polylog}}

\newcommand{\etal}{{et al. \xspace}}
\newcommand{\vol}{{\rm Vol}}
\newcommand{\Vol}{\vol}
\newcommand{\sk}{\mathsf{sk}}

\newcommand{\cC}{\mathcal{C}}
\newcommand{\cD}{\mathcal{D}}

\newcommand{\C}{\mathcal{C}}

\newcommand{\eps}{\epsilon}

\newcommand{\Czeta}{C}

\usepackage{pdfsync}

\title{Expander Decomposition in Dynamic Streams\thanks{The original version of this paper (as well as the version published in ITCS 23) claimed a construction of distance oracles in dynamic semi-streaming, improving over \cite{FKN21} distance oracle in the regime of $\tilde{O}(n^{1+\alpha})$ space (for a constant $\alpha\in(0,1)$). That construction contained a mistake, and was removed from the current version.}}
\author{Arnold Filtser\\Bar-Ilan University \and Michael Kapralov\\EPFL \and Mikhail Makarov\\EPFL}

\begin{document}
	\maketitle
	\thispagestyle{empty}
	\begin{abstract}

In this paper we initiate the study of expander decompositions of a graph $G=(V, E)$  in the streaming model of computation. The goal is to find a partitioning $\mathcal{C}$ of vertices $V$ such that the subgraphs of $G$ induced by the clusters $C \in \mathcal{C}$ are good expanders, while the number of intercluster edges is small. 
Expander decompositions are classically constructed by a recursively applying balanced sparse cuts to the input graph. In this paper we give the first implementation of such a recursive sparsest cut process using small space in the dynamic streaming model.

Our main algorithmic tool is a new type of cut sparsifier that we refer to as a power cut sparsifier -- it preserves cuts in any given vertex induced subgraph (or, any cluster in a fixed partition of $V$) to within a 
  $(\delta, \epsilon)$-multiplicative/additive error with high probability. The power cut sparsifier uses $\tilde{O}(n/\epsilon\delta)$ space and edges, which we show is asymptotically tight up to polylogarithmic factors in $n$ for constant $\delta$.

	\end{abstract}
	\newpage
	{\setcounter{tocdepth}{2} \tableofcontents}
	
	
	\newpage
	\pagenumbering{arabic}
	\setcounter{page}{1}

\section{Introduction}
Expanders are known to have many beneficial properties for algorithmic design. Therefore a natural divide-and-concur approach leads to breaking a given graph into expanders. An $(\eps,\phi)$-expander decomposition (introduced by \cite{GR99,KVV04}) of an $n$ vertex graph $G=(V,E)$ is a partition $\cC=\{C_1,\dots,C_k\}$ of the vertex set $V$ such that the conductance of each induced graph $G\{C_i\}$ is at least $\phi$, and such that there are at most $\eps\cdot|E|$ inter-cluster edges.
For every $\eps\in(0,1)$, one can always construct $\left(\eps,\phi\right)$-expander decomposition with $\phi=\Omega(\frac{\eps}{\log n})$, which is also the best possible (see, e.g., \cite{AALG18}).
Expander decomposition have been found to be extremely useful for  
Laplacian solvers \cite{SpielmanT04,CKP+17}, 
unique games \cite{ABS15,Trevisan08,RS10}, 
minimum cut \cite{KT19},
sketching/sparsification \cite{ACK+16,JS18,CGP+18},
 max flow algorithms \cite{KLOS14,CKLPPS22}, distributed algorithms \cite{CS20,CPSZ21,HRG22}
and dynamic algorithms \cite{NS17,Wul17,NSW17,BGS20,GRST21}.

Sequentially, expander decomposition can be constructed in almost linear time \cite{SW19,LS21}. Furthermore, expander decomposition have efficient ($\poly(\eps^{-1})\cdot n^{o(1)}$ rounds) constructions in distributed CONGEST model, and parallel PRAM and MPC models \cite{GKS17,chang2019improved,CS20}.
From the other hand, in the dynamic stream model (in fact even in insertions only model) it remained wide open whether it  is possible to construct an expander decomposition using $\wt{O}(n)$ space (which is the space required to store such a decomposition).

Graph sketching, introduced by~\cite{ahn2012analyzing} in an influential work on graph connectivity in dynamic streams, has been a de facto standard approach to constructing algorithms for dynamic streams, where the algorithm must use a small amount of space to process a stream that contains both edge insertions and deletions (while the vertex set is fixed). The main idea of~\cite{ahn2012analyzing} is to represent the input graph by its edge incident matrix, and to apply classical linear sketching primitives to the columns of this matrix. This approach seamlessly extends to dynamic streams, as by linearity of the sketch one can simply subtract the updates for deleted edges from the summary being maintained. A surprising additional benefit is the fact that such a sketching solution is trivially parallelizable: since the sketch acts on the columns of the edge incidence matrix, the neighborhood of every vertex in the input graph is compressed independently.

Sketching solutions have been constructed for many graph problems,
including spanning forest computation \cite{ahn2012analyzing}, 
spanner construction \cite{AGM12Spanners,KW14,FWY20,FKN21}, matching and matching size approximation \cite{AssadiKLY16,DBLP:conf/soda/AssadiKL17}, sketching the Laplacian \cite{ACK+16,JS18} and (among many other) most relevant to our paper, cut and spectral sparsifiers \cite{DBLP:conf/approx/AhnGM13,DBLP:conf/focs/KapralovLMMS14,DBLP:conf/soda/KapralovMMMNST20}.

\newcommand{\e}{\eps}

The main question we ask in this paper is: 
\begin{center}
What is the space complexity of constructing  expander decompositions in the sketching model?
\end{center} 
A classical way of constructing an expander decomposition (e.g.,~\cite{ST11}) is to recursively apply nearly balanced sparsest cuts as long as a cut of sparsity at most $\phi$ exists. The recursion terminates in $O(\log n)$ depth as long as one does not recurse on large sides of unbalanced cuts.
Indeed, if the graph does not include a sparse cut of balancedness more than $\frac14$, say, it could be shown that the large part induces an $\Omega(\phi)$-expander. Thus, one cuts $O(\phi \log n)$ fraction of edges over $O(\log n)$ levels of the recursion. Setting $\e=O(\phi \log n)$, we get  an $(\e, \Omega(\frac{\e}{\log n}))$-expander decomposition, which is existentially optimal.  The main question that our paper asks is

\begin{center}
\fbox{Can one design a space efficient streaming implementation of a \textbf{recursive} sparsest cut process?}
\end{center}

A natural approach to emulating this process in the streaming model of computation would be to run it on a sparsifier of the input graph. Alas, graph cut sparsifiers do not preserve cuts in induced subgraphs \footnote{Preserving cuts in subgraphs of arbitrary size, for example, implies preserving them in subgraphs induced by pairs of nodes -- and this requires preserving all the edges.}, so other methods are needed. Since we would like to implement the recursive sparsest cut process in the streaming model, at the very least we need to be able to know when to stop, i.e. be able to verify that the partition that we created is already an $(\e, \Omega(\frac{\e}{\log n}))$-expander decomposition. In particular, we ask
\begin{center}
Is it possible to verify, using small space, that a partition given at the end of the stream is an expander decomposition of the input graph?
\end{center}

Our {\bf first contribution} in this paper is a construction of what we call a power cut sparsifier -- a random (reweighted) subgraph of the input graph that can be constructed in small space, and can be used to verify the $\Omega(\frac{\e}{\log n})$-expansion property in any fixed partition (not known before the stream) with high probability. The power cut sparsifier contains $\widetilde{O}(n/\eps)$ edges, which we show is {\em optimal}. It is closely related to the $(\eps, \delta)$-cut sparsifier introduced in \cite{agarwal2022sublinear} and probabilistic spectral sparsifiers of \cite{lee2013probabilistic}. The main difference between them and the power cut sparsifier is that the latter provides cut preservation property with high probability for each possible subgraph of the original graph.

One might hope that the construction of an $\e$-power cut sparsifier allows for direct simulation of the above mentioned recursive sparsest cut algorithm of~\cite{ST11}: one simply uses a fresh power cut sparsifier for every level of the recursion, taking $\widetilde{O}(n/\eps)\cdot O(\log n)=\widetilde{O}(n/\eps)$ space overall. Surprisingly, this does not not work! The  issue is subtle: the termination condition that lets the~\cite{ST11} process not recurse on larger side of an unbalanced cut does not translate from the sparsifier to the actual graph, as it requires preservation of cuts in induced subgraphs -- see Section~\ref{sec:tech-exp-decomp} for a more detailed discussion.  
However, the more recent expander decomposition algorithm of Chang and Saranurak~\cite{chang2019improved} works for our purposes. It also uses the recursive balanced sparse cut framework and the ideas introduced in \cite{Wul17,NS17}. The termination condition is cut based, and translates from the sparsifier to the original graph. Since the depth of the recursion is $\widetilde{O}(1/\e$), this gives us an $(\e, \Omega(\frac{\e}{\log n}))$-expander decomposition using $\widetilde{O}(n/\e^2)$ space -- our {\bf second contribution} and the main result. We now state our results more formally:

\begin{restatable}[]{theorem}{ExpTime}
	\label{thm:main-exp}
	Given a dynamic stream containing a parameter $\eps \in (0, 1)$ and a graph $G=(V, E)$, there exists an algorithm that outputs a clustering $\mathcal{C}$ of $V$ that is a $(\eps, \Omega(\eps/\log n))$-expander decomposition of $G$ with high probability. This algorithms requires $\wt{O}(\frac{n}{\eps^2})$ memory and takes exponential in $n$ time.
\end{restatable}

Notice that the decomposition quality in \Cref{thm:main-exp} is asymptotically optimal (see, e.g., \cite{AALG18}).
Unfortunately, this algorithm takes exponential time due to the need to find balanced sparse cuts exactly. We also show how to make the runtime polynomial at the expense of slight losses in the quality of the expander decomposition and its space requirements:
%
%
%
%
%

\begin{restatable}[]{theorem}{PolyTime}
	\label{thm:main-poly}
	For an integer parameter $k \leq O(\log n)$, given a dynamic stream containing a parameter $\eps \in (0, 1)$ and a graph $G=(V, E)$, there exists an algorithm that outputs a clustering $\mathcal{C}$ of $V$ that is a $(\eps, \eps \cdot \Omega(\log n)^{-O(k)})$-expander decomposition of $G$ with high probability. This algorithms requires 
	\[
	\wt{O}(n)\cdot\left(\eps^{-2}+\eps^{\frac{1}{k}-1}\cdot n^{\frac{2}{k}}\cdot\log^{O(k)}n\right)
	\]
	memory and takes polynomial in $n, 1/\eps, (\log n)^k$ time.
\end{restatable}

\subsection{Power cut sparsifier}
A natural approach for constructing expander decomposition is to find the sparsest cut (or an approximation), remove its edges and recurse on  both sides. Once the sparsest cut found has conductance $\phi=\Omega(\frac{\eps}{\log n})$, the cluster at hand induces a $\phi$-expander and the algorithm halts. One can show that in this entire process at most an $\eps$-fraction of the edges is removed.


A cut sparsifier $H$ of $G$ (introduced by Benczur and Karger \cite{BK96}) with multiplicative error $1\pm\delta$ is a graph that preserves the values of all the cuts in $G$: $\forall S\subseteq V$, $\left|w_G(S,\overline{S})-w_H(S,\overline{S})\right|\le \delta\cdot w_G(S,\overline{S})$, where $w_G(S,\overline{S})$ is the sum of weights of all the edges crossing the cut $S$ (quantity in unweighted graphs).  In the dynamic stream model one can compute such a cut sparsifier using $\wt{O}(\frac{n}{\delta^2})$ space \cite{AGM12Spanners,DBLP:conf/soda/KapralovMMMNST20}.

We introduce a new type of sparsifier dubbed \emph{power cut sparsifier}. The name comes from the ambitious goal of sparsifying the entire power set of $V$. As this is impossible using the classical definition, we will also allow for an additive error. Denote by $\vol(S)=\sum_{v\in S}\deg(v)$ the sum of vertex degrees in $S$. 
$G\{S\}$ denotes the graph induced by $S$ with self loops (see \Cref{sec:prelims}).
\begin{restatable}[]{definition}{DefPowerCutSparsifier}
	\label{Def:PowerCutSparsifier}
	Consider a graph $G=(V,E_{G})$.
	A graph $H=(V,E_{H},w_{H})$ is an $(\delta,\epsilon)$-cut sparsifier
	of $G$ if for every cut $S\subseteq V$, 
	\[
	(1-\delta)\cdot w_{G}(S,\bar{S})-\epsilon\cdot\vol(S)\le w_{H}(S,\bar{S})\le(1+\delta)\cdot w_{G}(S,\bar{S})+\epsilon\cdot\vol(S)
	\]
	A distribution $\mathcal{D}$ over weighted subgraphs $H$ is called a $(\delta,\epsilon,p)$-power cut sparsifier distribution for $G$, if for every partition $\mathcal{C}$ of $G$, with probability at least $1-p$, $\forall C\in\mathcal{C}$, $H\{C\}$ is a $(\delta,\epsilon)$-cut sparsifier of $G\{C\}$. A sample $H$ from such a distribution is called a $(\delta,\epsilon,p)$-power cut sparsifier.
\end{restatable}

We show that power cut sparsifier can be constructed using very simple sampling algorithm: just add each edge $e=\{u,v\}$ to the sparsifier with probability $O(\frac{\log^2n}{\eps\cdot\delta})\cdot\left(\frac{1}{\deg(v)}+\frac{1}{\deg(u)}\right)$, see \Cref{thm:PowerCutSparsifier}.
Furthermore, we show that such a graph could be sampled in the dynamic stream model:
\begin{restatable}[]{theorem}{PowerCutSparsifierStream}
	\label{thm:pcs-alg}
	For every $C>0$, and $\eps,\delta\in(0,1)$, there exists an algorithm that uses $\tilde{O}(n\cdot \frac{\Czeta}{\epsilon \delta})$ space, that, given the edges of an $n$-vertex graph in a dynamic stream, in polynomial time produces samples from a $(\delta, \eps, n^{-\Czeta})$-power cut sparsifier distribution.
\end{restatable}
Given the simplicity of the sampling algorithm, and the very robust guarantee provided by power cut sparsifiers, we believe that further application (beyond the construction of expander decomposition) will be found, and hence think that they are interesting objects of study in their own right.
The dependence on $\eps$ in \Cref{thm:PowerCutSparsifier} (and hence also in \Cref{thm:pcs-alg}) is tight:
\begin{restatable}[]{theorem}{LBPowerCutSparsifier}
	\label{thm:cutlbPower}
	For every $\delta,\eps\in(\frac1n,\frac14)$ and even $n\in\N$, there is an $n$-vertex graph $G=(V,E)$ such that every $(\eps,\delta,\frac12)$-power cut sparsifier distribution $\cD$ produces graphs with $\Omega(\frac{n}{\eps})$ edges in expectation.
\end{restatable}
In our application of \Cref{thm:PowerCutSparsifier} we will use $\frac1\delta=O(\log n)$. Thus for our purpose here, the power cut sparsifier of \Cref{thm:PowerCutSparsifier} is tight up to polylogarithmic factors.
Furthermore, in the following theorem we show that  \Cref{thm:PowerCutSparsifier} is also tight (up to polylogarithmic factors) for the setting of parameters $\eps=\delta$.
In fact, this proposition holds even if one requires only a $(\delta,\eps)$-cut sparsifier, and not a power cut sparsifier, and even if we allow a general cut sketch (instead of a cut sparsifier). See \Cref{sec:LBbyCKST19} for definitions.

\begin{restatable}[]{theorem}{LBcutSparsifier}
	\label{thm:cutlb}
	For any $\epsilon,\delta\in(0,\frac14)$, any $(\delta,\epsilon)$-cut sketching scheme $\sk$ for unweighted multigraphs with $n$ vertices at most $n^2$ edges must use at least $\Omega(\frac{n\log n}{\max\{\eps^2,\delta^2\}})$ bits in the worst case.
\end{restatable} 

Note that the power cut sparsifier is closely related both to the $(\eps, \delta)$-cut sparsifier by \cite{agarwal2022sublinear} and the probabilistic spectral sparsifiers by \cite{lee2013probabilistic}. The former gives the same $(1 + \eps)$ multiplicative error on the size of cut $S$ and an additive error of $\delta \cdot |S|$. This is stronger than ours of $\eps \cdot \Vol(S)$, but only applicable when the graph doesn't have any self-loops, which is not the case for our applications. In other words, both constructions are tailored for their respective application and aren't applicable to the same extent in the other setting. Both constructions are achieved by sampling with a probability  based on the sum of inverse degrees of edge endpoints. In fact, the observation that sampling each edge with those probabilities produces a sparsifier was made by \cite{ST11}, Theorem 6.1. There, however, the original graph is required to be an expander. The subsequent works trade off this dependence for an additive error.

The simplest way to construct a power cut sparsifier is to virtually add an $\eps$-expander to the graph, and then apply the degree-based sampling scheme to the new graph. This results, however, in a quadratic dependence on $1/\eps$ in the number of edges. Reducing the dependence to just be linear is our main contribution to this line of work. A similar scheme with adding a constant degree expander is also proposed in \cite{agarwal2022sublinear}, and it can be shown that their construction can be extended to achieve construction of power cut sparsifier with the same dependance on $1/\eps$.

\subsection{Related work}\label{sec:related}
We refer the reader to \cite{McGregor14,McGregor17} for a survey of streaming algorithms. The idea of linear graph sketching was introduced in a seminal paper of Ahn, Guha, and McGregror \cite{ahn2012analyzing}. An extension of the sketching approach to hypergraphs were presented in~\cite{GuhaMT15}.  
%

Graph sparsifiers with additive error were previously studied by Bansal, Svensson, and Trevisan \cite{bansal2019new}. They studied sparsifiers where all edge weights are equal. As a result, their additive error dependence also on average degree (rather than only on volume). Another crucial difference is that our analysis provides sparsification guarantees for cuts for the majority of subgraphs of $G$.

\subsection{Preliminaries}\label{sec:prelims}

All of the logarithms in the paper are in base $2$. 
We use $\tilde{O}$ notation to suppress constants and poly-logarithmic factors in $n$, that is $\widetilde{O}(f)=f\cdot\polylog(n)$ (in particular, $\wt{O}(1)=\polylog(n)$).

Given an weighted graph $G=(V,E,w_G)$, the weighted degree $\deg_G(v)$ represents the sum of weights of all edges incident on $v$, including self-loops (simply degree for unweighted).
Given a set $S\subseteq V$, $G\{S\}$ denotes the graph induced by $S$ with self loops.
That is, the vertex set of $G\{S\}$ is $S$, we keep all the edges where both endpoints are in $S$, and we add self loops to every vertex such that its degree in $G$ and $G\{S\}$ is the same. We denote the complement set of $S$ by $\bar{S}=V\setminus S$.

The volume of a set $S$ is the sum of the weighted degrees of all the vertices in $S$: $\vol_G(S)=\sum_{v\in S}\deg_G(v)$. By $E_G(A, B)$ we represent the set of edges between vertex sets $A$ and $B$, and by $w_G(A, B) = \sum_{e\in E_G(A, B)}w_G(e)$ the sum of their weights (for unweighted their count).
Where the graph $G$ is clear from the context, we might abuse notation and drop the subscript (e.g. $\vol(S) := \vol_G(S)$).
We say that a set of vertices $S \subset V$ is a cut if $\emptyset \subsetneq S \subsetneq V$. The value $w_G(S, \bar{S})$ is called the size of the cut $S$, and the sparsity of a cut $S$ in graph $G$ is $\Phi_G(S) = \frac{w_G(S,\bar{S})}{\min\left\{ \vol_G(S),\vol_G(\bar{S})\right\}}$. We say that a cut $S$ is $\phi$-sparse if $\Phi_G(S) \leq \phi$.

\begin{definition}
	For $\phi \in [0, 1]$, a graph $G=(V, E,w_G)$ is a $\phi$-expander if the sparsity of every cut
	is at least $\phi$, i.e. $\min_{\emptyset \subsetneq S \subsetneq V}\Phi_G(S)\ge\phi$.    
\end{definition}

\begin{definition}
	For $\phi,\eps \in [0, 1]$ and a graph $G=(V,E)$, a partition $\mathcal{C}$ of the vertices $V$ is called an $(\epsilon,\phi)$-expander decomposition of $G$ if 
	\begin{itemize}
		\item the number of
		inter-cluster edges is at most $\epsilon$ fraction of all edges:
		\[
		\sum_{C\in\mathcal{C}} w(C,\bar{C}) \le \eps \cdot \vol(V)
		\]
		\item for each cluster $C \in \mathcal{C}$, $G\{C\}$ is a $\phi$-expander.    
		
	\end{itemize}
\end{definition}

\paragraph{Dynamic streams.} We say that a graph $G$ is given to us in a dynamic stream if we are given a fixed set $V$ of $n$ vertices and a sequence of updates on unweighted edges, where each update either adds an edge between two vertices or removes an existing edge. We define $G$ to be the graph at the end of the stream. 

\paragraph{Additive Chernoff bound.}
The following concentration bound is used in proof of the power cut sparsifier guarantees. Its proof appears in \Cref{appendix:AddChernoff}.
\begin{restatable}[Additive Chernoff Bound]{lemma}{AdditiveChernoff}\label{lem:add_chernoff_complete}
	Let $X_1, \ldots, X_n$ be independent random variables distributed in $[0, a]$. Let $X = \sum_{i\in [n]}X_i$, $\mu = \E[X]$. Then for $\eps \in (0, 1)$ and $\alpha \geq 0$
	\[
	\Pr[|X - \mu| > \eps \mu + \alpha] \leq 2 \exp \left(-\frac{\eps \alpha}{3a} \right)~.
	\]
\end{restatable}

\section{Technical overview} \label{sec:tech-overview}

\subsection{Verifying expander decompositions in a dynamic stream} Before explaining our algorithm for the expander decomposition construction, we first consider a problem of \emph{verifying} an expander decomposition. That is, given a graph $G$ in a dynamic stream and after it a clustering $\C$ of its vertices and the values $\eps$, $\phi$, we want to distinguish between the cases when $\C$ is an $(\eps, \phi)$-expander decomposition of $G$ and when the fraction of intercluster edges is at least $O(\eps)$ or some cluster contains a $O(\phi)$-sparse cut.

\paragraph{Using graph sparsifiers.} The core idea that allows us both to solve this problem and construct the expander decomposition using only limited memory is the use of the graph sparsifiers. According to the classical definition \cite{FHHP19}, the $\delta$-cut sparsifier of a graph $G=(V, E)$ is a graph $H=(V, E', w)$ such that, for any cut $\emptyset \subsetneq S \subsetneq V$
\[
    (1 - \delta) \cdot w_G(S, \bar{S}) \leq w_H(S, \bar{S}) \leq (1 + \delta) \cdot w_G(S, \bar{S}).
\]

It is well known that one can construct a cut sparsifier in the streaming setting \cite{AGM12Spanners} using $\wt{O}(n/\delta^2)$ space. Having such a sparisifier, we can already distinguish between the cases when $G$ is a $\phi$-expander and when there is a $\frac{1 - \delta}{1 + \delta}\phi$-sparse cut in $G$ just by checking all possible cuts in $H$.

Things become more involved when we try to verify an expander decomposition. Now, we need to check that for some partitioning of $V$ into clusters $\mathcal{C}$, all $G\{C\}$, $C \in \mathcal{C}$, are $\phi$-expanders. Here, the classical notion of the cut sparsifier is insufficient, as the approach of checking every cut in each of the $G\{C\}, C \in \mathcal{C}$ would require that we have a cut sparsifier for each $G\{C\}$.

\paragraph{Simple sparsifier for a partition.} A natural solution for this is to construct a cut sparsifier $H$ in such a way that for any fixed partition $\mathcal{C}$, $H\{C\}$ will be a cut sparsifier of $G\{C\}$ for all $C \in \mathcal{C}$ with high probability. 
From existing literature, we can already derive how to do that for graphs $G$ that are themselves $\phi$ expanders \cite{ST11}. Consider the following sampling procedure: the graph $G_{\phi,\delta}$ is obtained by
taking each edge $e=\left\{ u,v\right\} \in E$ 
with probability at least $p_{e}\ge\min\left\{ 1,\left(\frac{1}{\deg(u)}+\frac{1}{\deg(v)}\right)\cdot O\left(\frac{\log^{2}n}{\delta^{2}\phi^{2}}\right) \right\} $
and weight $\frac{1}{p_{e}}$. 
By Spielman and Teng \cite{ST11}, $G_{\phi,\delta}$ is sparsifier w.h.p.:
\begin{lemma}[\cite{ST11}]\label{lem:ST10}
	Let $G=(V,E)$ be a $\phi$-expander, then with probability $1-n^{\Omega(1)}$, $G_{\phi,\delta}$ is  $1+\delta$
	spectral sparsifier, and thus cut sparsifier of $G$.
\end{lemma}

A key observation is that if we were to apply this procedure to $G\{C\}$ for some cluster $C$, the sampling probabilities of each non-loop edge would remain the same since the degrees of vertices are unchanged. Because of that, with some additional work, one can show that $H\{C\}$ is a $(1 + \delta)$-cut sparsifier of $G\{C\}$ for each $C \in \mathcal{C}$ with high probability. 

\paragraph{Intuition behind power cut sparsifiers.} But how to adopt this approach to the case where $G$ is not an expander? We observe that for our purposes we can allow a small additive error in cut size estimation in terms of the volume of the cut. Namely, suppose that we have a sparsifier $H$ with the following guarantee: for any partition $\mathcal{C}$ of $V$ with high probability for any cut $\emptyset \subsetneq S \subsetneq V$,
\[
    (1 - \delta) \cdot w_{G\{C\}}(S, \bar{S}) - \eps \cdot \vol(S) \leq w_{H\{C\}}(S, \bar{S}) \leq 
    (1 + \delta) \cdot w_{G\{C\}}(S, \bar{S}) + \eps \cdot \vol(S).
\]
See \Cref{Def:PowerCutSparsifier}.
This sparsifier is suitable for checking whether each $G\{C\}$ is an expander, since this guarantee implies the following guarantee for expansion of the cut $S$:
\[
    (1 - \Theta(\delta))\Phi_{G\{C\}}(S) - \Theta(\eps) \leq \Phi_{H\{C\}} \leq (1 + \Theta(\delta))\Phi_{G\{C\}}(S) + \Theta(\eps).
\]

Since we only want to check whether the expansion is bigger than $\phi$ with constant precision, it is enough for $\delta$ to be a small constant and $\eps = O(\phi)$. It turns out that we can construct such a sparsifier, which we call a power cut sparsifier, by sampling each edge with probability 
$p_{e}\ge\min\left\{ 1,\left(\frac{1}{\deg(u)}+\frac{1}{\deg(v)}\right)\cdot O\left(\frac{\log^{2}n}{\delta \eps} \right) \right\}$. This results in the sparsifier having size $O \left(\frac{n \log^2 n}{\delta \eps} \right)$, featuring only \emph{linear} dependance on $1/\eps$. This dependance is optimal, as we show in \Cref{thm:cutlbPower}.

\paragraph{Showing correctness.} 
Our proof that such a sampling scheme indeed produces a power cut sparsifier is based on the idea of decomposing cuts into $d$-projections, introduced by Fung \etal \cite{FHHP19}.
While \cite{FHHP19} work with sampling probabilities based on edge connectivities, in our case the probabilities only depend on the degrees of the vertices. 
We define the $d$-projection of a cut $S$ to be the subset of edges of $E(S,\bar{S})$ who's both endpoint have degree at least $d$, and show  that the overall number of distinct $d$-projections is polynomially bounded (\Cref{clm:projections}). The rest of the proof is standard: we show that the value of each $d$-projection is preserved with high probability by using an multiplicative-additive version of Chernoff bound, and then take a union bound over all projections.

\subsection{Constructing an expander decomposition} \label{sec:tech-exp-decomp}
\paragraph{Naive approach.} Having seen that we can verify an expander decomposition using a power cut sparsifier, a question arises: can we construct an expander decomposition provided with only power cut sparsifiers? A natural approach would be to construct a decomposition of the sparsifier using one of the existing methods and then to claim that it is also a decomposition of the original graph. Unfortunately, this plan is flawed, because the resulting decomposition would \emph{depend} on the used power cut sparsifier, which violates the implicit assumption that the partition $\mathcal{C}$ is chosen \emph{independently} from the sparsifier. Hence we cannot claim that the produced clustering is an expander decomposition. One can attempt to verify that using a fresh copy of the power cut sparsifier. However, if the verification fails were are back at square one.

\paragraph{Recursive approach.}
To explain our algorithm let's first examine the classical recursive approach to expander decomposition construction \cite{ST13}. The simplest version is as follows: find a sparse cut $S$, make this cut and then recurse on both sides $S$, $V \setminus S$. As we have seen, we cannot run all if those steps on one sparsifier, so we would use a fresh power cut sparsifier for each level of recursion.

Apriori, one of the sides of the sparse cut might be very small, implying that the recursion depth of this algorithm might be even of linear size, resulting in a quadratic space requirement overall.
Instead, we employ a classic approach of Spielman and Teng \cite{ST13}. The idea is to use a balanced sparse cut procedure to find a balanced sparse cut in the sparsifier, i.e. a cut that is sparse and both sides of which have approximately the same volume. If we use this procedure instead of just finding an arbitrary sparse cut, we can guarantee that the recursion depth is only logarithmic. However, there are a couple caveats that we discuss next.

\paragraph{Balanced sparse cut procedure.} A balanced sparse cut procedure typically takes the sparsity $\phi$ as an input and has three possible outcomes:
\begin{itemize}
    \item It declares that the graph is a $O(\phi)$-expander.
    \item It finds a balanced $O(\phi)$-sparse cut. A cut is balanced if both of its sides have approximately the same size.
    \item It finds an unbalanced $O(\phi)$-sparse cut, but gives some kind of guarantee on the largest side of the cut. For example, it may guarantee that it is a $O(\phi)$-expander.
\end{itemize}

If one uses it in the recursive algorithm outlined above, first outcome means that we can stop recursing on the graph. The second means that the cut we make is balanced, so we make the cut and recurse. Finally, in the last outcome because the bigger side of the cut is $O(\phi)$-expander, we can recurse only on the smaller side of the cut. This means that when we recurse, the volume of the part that we recurse on always decreases by a constant fraction, so the depth of the recursion is at most logarithmic.

In our case, we run it on a sparsifier $H$ and are interested in implications for the original graph $G$. The first two outcomes translate from a sparsifier to the original graph with only a small error, as they both concern only the cuts in the graph. The last outcome suffers the most: it concern the cut structure of some subgraph $H\{C\}$ of $H$. But because the choice of $C$ is dependent on $H$, we cannot guarantee that $H\{C\}$ is a sparsifier of $G\{C\}$. Hence the guarantee that $H\{C\}$ is an expander only translates into the guarantee that $G\{C\}$ is a near-expander, which is a much weaker guarantee.
This motivates us to consider an expander decomposition that uses a version of balanced sparse cut procedure with a relatively weak third guarantee, but one which survives this transition. We opt to adopt the approach of Chang and Saranurak 
\cite{chang2019improved}.

\paragraph{Two phase approach.} 
The weak third guarantee makes bounding the depth of the recursion difficult. That is why the algorithm is split into two phases. The first phase (\Cref{alg:streaming_ed}) is recursive: given a vertex set $C$, we run the most balanced sparse cut procedure on sparsifier $H\{C\}$. If it declares that it is an expander, then the graph $G\{C\}$ is also an expander and we terminate. If there is a balanced sparse cut $S$, we make this cut and recurse on vertex sets $S$ and $\bar{S}$ with a fresh sparsifier. Otherwise, if the cut is more than $O(\phi)$ unbalanced we execute second phase on $C$.

The idea behind the second phase (\Cref{alg:decomp_unbal}) is that if $G\{C\}$ doesn't have a $O(\phi)$-balanced sparse cut for sparsity $\phi$, then if we were to successively make cuts of sparsity less than $\phi$, we could not cut more than $O(\phi)$ fraction of the total volume of $C$. Because we have a small error in approximating the cut size and, therefore, the cut sparsity, due to the sparsifier usage, we decrease the sparsity of the cut that we are looking for by multiplying it by a close to $1$ constant $c < 1$ to mitigate that error. Furthermore, we divide the required cut balancedness by a constant. We use the balanced sparse cut procedure with sparsity $c \phi$ and continue in the same manner as in phase one, except that if there is a balanced sparse cut, we split the smaller side of the cut into singletons and continue with the larger. Finally, when we find an unbalanced cut, we use the same argument and decrease $\phi$ and balancedness even further. Eventually, the balancedness threshold becomes so small that all cuts are trivially balanced, and thus the algorithm halts and declare the cluster to be an expander.

\subsection{Roadmap}

The rest of the paper is organized as follows. First, in \Cref{sec:powerCutSSparsifier}, we give the algorithm for constructing a power cut sparsifier in RAM computation model and show its correctness. On top of that, in this section we also discuss the construction in dynamic streams, as well as give lower bounds on the expected size of the power cut sparsifier and the size of a $(\delta, \epsilon)$-cut sparsifier.

In \Cref{sec:balancedSparseCut}, we show that existing algorithms for finding a balanced sparse cut can be successfully used on a cut sparsifier instead of the original graph.

\Cref{sec:expanderDecompositionProof} presents our algorithm and its correctness proof, culminating in \Cref{thm:main-exp} and \Cref{thm:main-poly}.


\section{Power Cut Sparsifier} \label{sec:powerCutSSparsifier}
We would like to construct a cut sparsifier $H$ of a graph $G$ such
that for every cluster $C \subseteq V$, $H\{C\}$ will be a cut sparsifier of $G\{C\}$ (w.h.p.). 
Under standard $1\pm\epsilon$ multiplicative error, such sparsifier would have to contain all of the edges $G$, as for every edge $(u,v)\in E$ we can pick $C=\{u,v\}$.

For this reason, we relax the requirement in two ways: first, we will allow additive
error. Secondly, we will only require this property to hold for sets $C$ of a partition $\mathcal{C}$ that we can fix post factum. We restate \Cref{Def:PowerCutSparsifier}  for convenience.
\DefPowerCutSparsifier*

The name ``power cut sparsifier'' indicates that $H$ sparsifies the majority of the subgraphs induced by the power set of vertices of $G$.
Previous constructions \cite{BK96,ST11,BSS14} provided $(\delta,0)$-cut sparsifier with $\wt{O}(\frac{n}{\delta^2})$ edges. However, no guarantee on all the subsets was previously known. 

Sampling from $\mathcal{D}$, which we will now present, is achieved through standard techniques for sparsifier construction.

\paragraph{Algorithm.} 
Consider a graph $G=(V,E)$ and parameter $\Upsilon=\frac{6(\Czeta+2)}{\delta\cdot\epsilon}\cdot2\log n\cdot\ln n=O(\frac{\Czeta\log^{2}n}{\epsilon\cdot\delta})$.
For every edge $e=(u,v)$, let $1 \geq p_{e}\ge\min\left\{ 1,\Upsilon\cdot\left(\frac{1}{\deg(u)}+\frac{1}{\deg(v)}\right)\right\}$ be some arbitrarily determined probability.
Construct a graph $H$ by adding every edge $e\in E$ independently with probability $p_e$, and giving it weight $\frac{1}{p_e}$ in case it was sampled. 

\begin{restatable}[]{theorem}{PowerCutSparsifier}
	\label{thm:PowerCutSparsifier}
	The distribution over weighted subgraphs described above is a $(\delta,\epsilon,n^{-\Czeta})$-power cut sparsifier distribution. Furthermore, with probability at least $1 - n^{-\Czeta}$ the number of edges in $H$ is at most $O(n\cdot\frac{\Czeta\cdot\log^2 n}{\eps\cdot\delta})$.
\end{restatable}
\begin{proof}
	By linearity of expectation, the expected number of edges is bounded by
	\begin{align*}
		\mathbb{E}\left[|E(H)|\right] & =\sum_{e=\{u,v\}\in E}p_{e}\le\Upsilon\sum_{e=\{u,v\}\in E}\left(\frac{1}{\deg(u)}+\frac{1}{\deg(v)}\right)\\
		& =\Upsilon\cdot\sum_{v\in V}\sum_{u\in N(v)}\frac{1}{\deg(v)}=n\Upsilon=O\left(n\cdot\frac{\Czeta\cdot\log^2 n}{\eps\cdot\delta}\right).
	\end{align*}
	By \Cref{lem:add_chernoff_complete}, 
	\[
		\Pr\big(|E(H)| \geq 2 n \Upsilon\big) \leq \Pr\big(|E(H)| \geq \E[|E(H)|]  + n \Upsilon\big)
		\leq 2\cdot e^{\frac{-n \Upsilon}{3}} \leq n^{-\Czeta}~,
	\]
	which concludes the proof of the second assertion of the theorem.

	Now we show the first assertion of the theorem.
	Our proof is based on projections, which is inspired by the sparsifier construction of Fung \etal \cite{FHHP19}.
	The \emph{degree} of an edge $e=\{u,v\}$ is set to be $\deg(e)=\min\left\{ \deg(u),\deg(v)\right\}$, i.e. the smaller degree between its adjacent vertices.

	\begin{definition}
		For a cut $\emptyset \subsetneq S \subsetneq V$, the $d$-projection of $S$ is the set of edges from $E(S,\bar{S})$ with degree at least $d$. 
	\end{definition}
	We begin by bounding the number of different $d$-projections. Recall that $\vol(S)=\sum_{v\in S}\deg(v)$.
	\begin{lemma}\label{clm:projections}
		For every integer $d,\alpha\ge1$, the number of distinct $d$-projections in cuts of volume at most $\alpha\cdot d$ is at most $n^{\alpha}$.
	\end{lemma}	
	\begin{proof}
		Let $V_{d}=\left\{ v\in V\mid\deg(v)\ge d\right\} $ be all the vertices
		of degree at least $d$. Note that edges of degree $d$ are only edges
		between vertices in $V_{d}$. 
		Consider a cut $S$, note that any addition
		and removal of vertices from $V\setminus V_{d}$ from $S$ will not
		change the $d$-projection.
		Thus the $d$ projection of $S$ is determined by $S\cap V_{d}$.
		Note also that each vertex in $S\cap V_{d}$ contributes at least $d$
		to the volume of $S$. If $\vol(S)\le\alpha d$, then $|S\cap V_d|\le\alpha$.
		We conclude that the number of distinct $d$-projections in cuts of volume at most $\alpha\cdot d$
		is bounded by $\sum_{i=0}^{\alpha}{|V_{d}| \choose i}\le\sum_{i=0}^{\alpha}\frac{n^{i}}{2}\le n^{\alpha}$.
	\end{proof}


Set $F_{i}=\left\{ e\in E\mid\deg(e)\in[2^{i},2^{i+1})\right\}$,
and let $G_{i}=(V,F_{i})$ denote the subgraph of $G$ containing edges in $F_i$.
Let $H_{i}$ be the subgraph of $H$
restricted to the edges of $F_{i}$. Note that for every cut $S$,
$\mathbb{E}\left[w_{H_{i}}(S,\bar{S})\right]=w_{G_{i}}(S,\bar{S})$.
Let $\epsilon'=\frac{\epsilon}{2\log n}$, thus $\Upsilon=\frac{6(\Czeta+2)}{\delta\cdot\epsilon'}\cdot\ln n$.
\begin{lemma}\label{lem:fixed-i-all-cuts}
	For every fixed $i$, with probability at least $1-4\cdot n^{-2(\Czeta+1)}$,
	\[
	\forall S\subseteq V,\left|w_{H_{i}}(S,\bar{S})-w_{G_{i}}(S,\bar{S})\right|<\delta\cdot w_{G_{i}}(S,\bar{S})+\epsilon'\cdot\vol(S)
	\]
\end{lemma}
\begin{proof}
	For every cut $S$ in $G$ of volume smaller than $<2^{i}$, all the vertices
	in $S$ are of degree smaller than $2^{i}$, and hence $w_{G_{i}}(S,\bar{S})=\emptyset$.
	Thus the lemma holds on all these cuts, and we will consider only cuts of volume at least $2^i$. If $\Upsilon\ge2^{i+1}$,
	all the edges in $F_{i}$ belong to $H$, and hence $w_{H_{i}}(S,\bar{S})=w_{G_{i}}(S,\bar{S})$
	for every cut $S$. Hence we will assume that $\Upsilon<2^{i+1}$.
	
	Fix a cut $S$. For an edge $e\in E_{G_{i}}(S,\bar{S})$,
	let $X_{e}$ be a random variable getting the weight of the edge $e$
	in $H$. For every such $e$, $\deg_{G}(e)\in[2^{i},2^{i+1})$ and hence
	$p_{e}\ge\frac{\Upsilon}{2^{i+1}}$. In particular, $X_{e}\in[0,\frac{2^{i+1}}{\Upsilon}]$.
	Using \Cref{lem:add_chernoff_complete}, we conclude
	\begin{align}
		& \Pr\left[\left|w_{H_{i}}(S,\bar{S})-w_{G_{i}}(S,\bar{S})\right|\ge\delta\cdot w_{G_{i}}(S,\bar{S})+\epsilon'\cdot\vol(S)\right]\nonumber\\
		& \qquad=\Pr\left[\left|\sum_{i}X_{i}-\mathbb{E}[\sum_{i}X_{i}]\right|\ge\delta\cdot\mathbb{E}[\sum_{i}X_{i}]+\epsilon'\cdot\vol(S)\right]\nonumber\\
		& \qquad\le2\cdot e^{-\frac{\delta\cdot\epsilon'\cdot\vol(S)}{3\cdot\frac{2^{i+1}}{\Upsilon}}}=2n^{-\frac{\Czeta+2}{2^{i}}\cdot\vol(S)}~.\nonumber
	\end{align}
	Let $\mathcal{S}_{i,j}=\left\{ S\subseteq V\mid\vol(S)\in[2^{i+j},2^{i+j+1})\right\} $
	be all the cuts in $G$ with volume between $2^{i+j}$ and $2^{i+j+1}$, and
	let $\widehat{\mathcal{S}_{i,j}}\subseteq\mathcal{S}_{i,j}$ be a
	subset of them such that every $2^{i}$-projection in $G$ of a cut in $\mathcal{S}_{i,j}$ is represented
	by some cut in $\widehat{\mathcal{S}_{i,j}}$. 
	Then any $2^i$-projection of any non-empty cut in $G_i$ is represented by some cut in $\cup_{j\ge1}\widehat{\mathcal{S}_{i,j}}$.
	Hence, it is enough to prove
	the lemma on all the cuts in $\cup_{j\ge1}\widehat{\mathcal{S}_{i,j}}$. Using 
	\Cref{clm:projections}, we get that $\left|\widehat{\mathcal{S}_{i,j}}\right|\le n^{2^{j+1}}$.
	This together with the union bound gives
	\begin{align*}
		& \Pr\left[\exists S\subseteq V,\left|w_{H_{i}}(S,\bar{S})-w_{G_{i}}(S,\bar{S})\right|\ge\delta\cdot w_{G_{i}}(S,\bar{S})+\epsilon'\cdot\vol(S)\right]\\
		& \qquad\le\sum_{j\ge1}\sum_{S\in\widehat{\mathcal{S}_{i,j}}}\Pr\left[\left|w_{H_{i}}(S,\bar{S})-w_{G_{i}}(S,\bar{S})\right|\ge\delta\cdot w_{G_{i}}(S,\bar{S})+\epsilon'\cdot\vol(S)\right]\\
		& \qquad\le\sum_{j\ge1}^{n}\sum_{S\in\widehat{\mathcal{S}_{i,j}}}2n^{-\frac{\Czeta+2}{2^{i}}\cdot\vol(S)}<2\sum_{j\ge1}^{n}n^{2^{j+1}}n^{-\frac{\Czeta+2}{2^{i}}\cdot2^{i+j}}=2\sum_{j\ge1}^{n}n^{-(\Czeta+1)\cdot2^{j}}<4\cdot n^{-2(\Czeta+1)}~.
	\end{align*}
\end{proof}
By union bound, as the number of different possible degrees is at most $n$,
with probability $1-n^{-\Czeta}$ the lemma holds for
all the indices $i$. 
Conditioning on the success event, we get for every
cut $S$
\begin{align}
	\left|w_{H}(S,\bar{S})-w_{G}(S,\bar{S})\right| & \le\sum_{i=0}^{\lfloor\log n\rfloor}\left|w_{H_{i}}(S,\bar{S})-w_{G_{i}}(S,\bar{S})\right|\nonumber\\
	& \le\sum_{i=0}^{\lfloor\log n\rfloor}\left(\delta\cdot w_{G_{i}}(S,\bar{S})+\epsilon'\cdot\vol(S)\right)\nonumber\\
	& \le\delta\cdot w_{G}(S,\bar{S})+\epsilon'\cdot\log2n\cdot\vol(S)=\delta\cdot w_{G}(S,\bar{S})+\epsilon\cdot\vol(S)~.\label{eq:sumAllScaled}
\end{align}
We conclude that with probability $1-n^{-\Czeta}$, the graph $H$ is an $(\delta,\eps)$-cut sparsifier of $G$.
 
Consider a partition $\mathcal{C}$ of $V$. The graph obtained by taking the union of  $\bigcup\left\{H\{C\}\right\}_{C\in\mathcal{C}}$ is a sample of $\bigcup\left\{G\{C\}\right\}_{C\in\mathcal{C}}$ according to our algorithm. 
Note that our analysis allows for self loops (contributing only $1$ edge to the degree) and multi-edges.

Thus with probability $1-n^{-\Czeta}$, $\bigcup\left\{H\{C\}\right\}_{C\in\mathcal{C}}$ is an $(\delta,\eps)$-cut sparsifier of $\bigcup\left\{G\{C\}\right\}_{C\in\mathcal{C}}$.
Note that for every $C'\in\mathcal{C}$, and $S\subseteq C'$, $w_{\bigcup\left\{G\{C\}\right\}_{C\in\mathcal{C}}}(S,V\setminus S)=w_{G\{C\}}(S,C\setminus S)$. If follows that for every $C\in\mathcal{C}$, $H\{C\}$ is a $(\delta,\eps)$-cut sparsifier of $G\{C\}$. Hence, our distribution is a $(\delta,\eps,n^{-\Czeta})$-power cut sparsifier of $G$.

\end{proof}

\subsection{Dynamic Stream implementation}

In order to construct a power cut sparsifier, it is enough to independently sample each edge $e$ with probability $p_e$ and, if it is sampled, assign it the weight $1/p_e$. To this end, we adopt sampling strategy proposed in \cite{AGM12Spanners,ahn2012analyzing}. We first observe that we can make an assumption that we have access to $\log n$ uniform hash function $h_{i}: {V \choose 2} \to \{0, 1\}$.
Uniform hash simply means that we sample uniformly at random one of the $2^{{n\choose 2}}$ functions from ${V \choose 2}$ to $\{0, 1\}$. Even though representing such function apriori requires $\Omega(n^2)$ space, we will later use Nisan pseudorandom generators \cite{nisan1992pseudorandom} to reduce the space requirement to $\wt{O}(n \Upsilon)$. 

The hash functions induce $\log n$ graphs $G_i=(V, E_i)$ for $i\in I=\{1,\dots,\log n\}$, where $e \in E_i$ iff $e \in E$ and $\prod_{j \leq i} h_j(e) = 1$. That means that each edge from $G$ belongs to $G_i$ with probability $2^{-i}$ (independently of all other edges). These hush functions naturally induced by $n^2\log n$ random bits.


While storing the graphs $G_i$ without information loss is infeasible, for each vertex we will be interested in recovering its adjacency vector in only one $G_i$ where its expected degree is low. This can be achieved through a folklore exact sparse recovery algorithm. We take its formulation from \cite{AGM12Spanners}.

\begin{lemma}[Folklore, sparse-recovery] \label{lem:exact-recovery}
    There exists a sketch-based algorithm $k$-Recovery that recovers a $\{0, 1\}^n$-vector $x$ from a dynamic stream exactly with probability at least $1 - p$ if $x$ has only $k$ non-zero entries and outputs $FAIL$ otherwise. It uses $O(k \cdot \log \frac{1}{p}\cdot\log n)$ space.
\end{lemma}

Our algorithm will work as follows: for every vertex $v$ during the dynamic stream we will maintain a simple counter $\deg(v)$ for its degree, and in addition for every $i\in I$, a sketch $g_{i, v}$ for the neighborhood of $v$ in $G_i$, using \Cref{lem:exact-recovery} with parameters $k=8\Upsilon$ and $p=n^{-\Czeta - 3}$.
Note that the total space requirement is $\wt{O}(n)$. 

For each vertex $v$, let $j_v=\max\left\{ 0,\left\lfloor \log\left(\frac{\deg(v)}{2\Upsilon}\right)\right\rfloor \right\}$. 
Let $N_{G_{j_v}}(v)$ be the set of all neighbors of $v$ in $G_{j_v}$. Let $G'=\cup_{v\in V}\left(v\times N_{G_{j_v}}(v)\right)$ be the graph obtained by taking all the edges incident on $v$ in $G_{j_v}$ for every $v$.
For every recovered edge $e=\{u,v\}$, set it's weight in $G'$ to be $w_e=2^{\min\{j_v,j_u\}}$.
We can recover $G'$: for every $v$ using $g_{j_v,v}$ we recover $N_{G_{j_v}}(v)$ and construct $G_{\rm rec}$ as the union of all the recovered edges (with the appropriate weights). Finally, we return $G_{rec}$ as the answer.

We begin by arguing that the distribution over the graphs $G'$ is indeed a power cut sparsifier. Later, we will argue that we successfully recover $G'$ with high probability.
\begin{lemma}
	Every edge $e=\{u, v\} \in E$ belongs to $G'$ with probability 
	\[
		p_{e} = 2^{-\min\{ j_v, j_u \}} \ge\min\left\{ 1,\Upsilon\cdot\left(\frac{1}{\deg(u)}+\frac{1}{\deg(v)}\right)\right\},
	\]
	independently from the other edges. Further, if $e$ is sampled, it's weight is $p_e^{-1}$.
\end{lemma}
\begin{proof}
	Consider an edge $e=\{u,v\}\in E$, where w.l.o.g. $\deg(v)\le\deg(u)$. 
	If $\deg(v)<4\Upsilon$, then $j_v=0$. In particular, $u\in N_{G_0}(v)$, and thus $e\in G'$ with weight $w_e=2^{\min\{j_v,j_u\}}=1$.
	
	Otherwise, as  $G_{j_u}\subseteq  G_{j_v}$, $e\in G'$ if and only if $e\in G_{j_v}$. This happens with probability $p_e=2^{-j_{v}}=2^{-\left\lfloor \log\left(\frac{\deg(v)}{2\Upsilon}\right)\right\rfloor }\ge\frac{2\Upsilon}{\deg(v)}\ge\Upsilon\cdot\left(\frac{1}{\deg(u)}+\frac{1}{\deg(v)}\right)$. If it is indeed sampled, it's weight is set to be $w_e=2^{\min\{j_v,j_u\}}=2^{j_v}=p_e^{-1}$.
\end{proof}

By \Cref{thm:PowerCutSparsifier} it follows that the distribution producing $G'$ is a $(\delta,\eps,n^{-\Czeta})$-power cut sparsifier.
Next we argue that $G_{\rm rec}$ is very likely to be equal to $G'$.
Note that the number of edges in $G_{\rm rec}$ is always bounded by $nk=O(n \Upsilon)$.
Denote by $\Psi$ the event that $G_{\rm rec}\ne G'$, i.e. we failed to recover $G'$.
\begin{lemma} \label{lemma:pcs-alg-correctness}\label{lem:RecoveryProb}
$\Pr[\Psi]\le  n^{-\Czeta}$.
\end{lemma}
\begin{proof}
	For a successful recovery, it is enough that the degree of every vertex $v$ in $G_{j_v}$ is bounded by $k=8\Upsilon$, and that all the sparse recovery sketches $\{g_{j_v,v}\}_{v\in V}$ succeeded. 
	
	For every vertex $v$, if $j_v = 0$, then  $\log\left(\frac{\deg(v)}{2\Upsilon}\right)<1$, implying that
	$|N_{G_{0}}(v)|=\deg(v)\le2\Upsilon<k$. 
	Otherwise, it holds that  $\frac{2\Upsilon}{\deg(v)}\le2^{-j_{v}}<\frac{4\Upsilon}{\deg(v)}$.
	As each edge from $N(v)$ is present in $N_{j_v}(v)$ with probability $2^{-{j_v}}$,
	$\E[|N_{j_v}(v)|] \leq 4 \Upsilon$. Hence by \Cref{lem:add_chernoff_complete}, 
	\begin{align*}
	\Pr[|N_{j_{v}}(v)| & \geq8\Upsilon]\leq\Pr\left[\left||N_{j_{v}}(v)|-\E[|N_{j_{v}}(v)|]\right|\geq\frac{1}{2}\cdot\E[|N_{j_{v}}(v)|+2\Upsilon\right] \leq2e^{-\frac{\Upsilon}{3}}=2n^{-\frac{4(\Czeta+2)}{\delta\cdot\epsilon}\cdot\log n}~.
	\end{align*}
	By union bound (and using \Cref{lem:exact-recovery}), the probability that either $|N_{j_{v}}(v)|  \geq8\Upsilon$ for some $v$, or some $\{g_{j_v,v}\}_{v\in V}$ is failed recovery is bounded by 
	\[
	n\cdot\left(2n^{-\frac{4(\Czeta+2)}{\delta\cdot\epsilon}\cdot\log n}+n^{-\Czeta-3}\right)\le n^{-\Czeta}~.
	\]
%
\end{proof}

Finally, we no 	te that the number of required random bits can be reduced to nearly linear in $n$ using Nisan's pseudorandom generator~\cite{nisan1992pseudorandom} (at the expense of a factor $\Omega(n)$ increase in runtime per update), or the faster pseudorandom number generator of~\cite{DBLP:conf/soda/KapralovMMMNST20,DBLP:journals/corr/abs-1903-12165} that can be applied at polylogarithmic cost per evaluation.   This gives

\PowerCutSparsifierStream*

\subsection{Lower bound for power cut sparsifier}
In this section we prove \Cref{thm:cutlbPower}, restated bellow for convenience.
\LBPowerCutSparsifier*
\begin{proof}
	Set $d=\left\lfloor \frac{1}{2\eps}\right\rfloor$, and let $G=(V,E)$ be a $d$-regular graph. Let $\cD$ be an $(\eps,\delta,\frac12)$-power cut sparsifier.
	Consider an edge $e=\{u,v\}\in E$, and let $C=\{u,v\}$ be a set, and $\cC_e=\{C,V\setminus C\}$ be a partition. The graph $G\{C\}$ consist of two vertices $u,v$ with a single edge between them, and $d-1$ self loops on each vertex.	
	Let $S=\{u\}$, $\overline{S}=\{v\}$. Consider a graph $H$ drawn from $\cD$.
	With probability at least $1-\frac12=\frac12$, it holds that,
	\[
	w_{H\{C\}}(S,\overline{S})\ge(1-\delta)\cdot w_{G\{C\}}(S,\bar{S})-\epsilon\cdot\vol(S)=(1-\delta)-\eps\cdot d>0~.
	\]
	It follows that $\Pr_{H\sim\cD}[e\in H]\ge \frac12$ (as otherwise $	w_{H\{C\}}(S,\overline{S})=0$). We conclude that $\mathbb{E}_{H\sim\cD}[|H|]\ge\frac{1}{2}\cdot|E|=\frac{dn}{2}=\Omega(\frac{n}{\eps})$, as required.	
\end{proof}

\subsection{Lower bound for $(\delta,\eps)$-cut sparsifier}\label{sec:LBbyCKST19}
%
%
%
%
%
%
%
Given a graph $G=(V,E)$, an $(\delta,\epsilon)$-cut sketching scheme $\sk$ is a function $\sk:2^V\rightarrow\R_{\ge0}$, such that for every subset $S\subseteq V$ it holds that
\[
(1-\delta)\cdot w_{G}(S,\bar{S})-\epsilon\cdot\vol_{G}(S)\le\sk(S)\le(1+\delta)\cdot w_{H}(S,\bar{S})+\epsilon\cdot\vol_{H}(S)~.
\] 
Note that, given a $(\delta,\epsilon)$-cut sparsifier $H$, one can create a cut sketching scheme by setting $sk(S)=w_{H}(S,\bar{S})$. Our lower bound in  \Cref{thm:cutlb} holds for arbitrary sketching scheme. In particular, it implies that the number of edges in a cut sparsifier must be  $\Omega(\frac{n}{\max\{\eps^2,\delta^2\}})$.
It also follows that our construction in \Cref{thm:PowerCutSparsifier} is tight for the regime $\eps=\delta$.
 We restate the theorem for convenience.

\LBcutSparsifier*
We deffer the proof to the full version of this paper, and provide here only a sketch.
\begin{proof}[Proof sketch.\\]
	Fix $d=\Theta(\frac{1}{(\eps+\delta)^2})$. Denote by $\mathcal{G}_{n,n,d,\phi}$ the family of $n,n$-bipartite $d$-regular multigraphs (with fixed sides $L,R$) which are $\phi=\Theta(1)$-expanders. 	
	We first argue that $|\mathcal{G}_{n,n,d,\phi}|\ge 2^{\alpha\cdot dn\log n}$ for some constant $\alpha$. This could be shown by taking $d$ random matchings between $L$ and $R$ (there are about $(n!)^d$ graphs in the support), and arguing that w.h.p. the resulting graph is $\phi$-expander.
	
	Next, following the rigidity of cut approximation lemma of Carlson, Kolla, Srivastava, and Trevisan \cite{CKST19}, we get that if for two multigraphs $H,G\in \mathcal{G}_{n,n,d,\phi}$ it holds that $H$ is $(\delta,0)$-cut sparsifier of $G$, then   $G$ and $H$ must have at least $dn\cdot(1-\Theta(\sqrt{d})\delta)$ edges in common. \footnote{The proof in \cite{CKST19} holds for simple bipartite $d$-regular graphs, but can be easily extended to $\mathcal{G}_{n,n,d,\phi}$.}
	
	 Observe that if $H$ is a $(\delta,\eps)$-sparsifier of a $\phi$-expander $G$, then for every cut $S$ it holds that 
	 \[
	 \left|w_{G}(S,\bar{S})-w_{H}(S,\bar{S})\right|\le\delta\cdot w_{G}(S,\bar{S})+\epsilon\cdot\vol(S)\le(\delta+\frac{\eps}{\phi})\cdot w_{G}(S,\bar{S})~,
	 \]
	 and thus $H$ is a $(\delta+\frac\eps\phi,0)$-cut sparsifier of $G$, and in particular by \cite{CKST19}, if $H,G\in \mathcal{G}_{n,n,d,\phi}$ then they have $dn\cdot(1-\Theta(\sqrt{d})(\delta+\frac\eps\phi))=\Omega(dn)$ edges in common.
	 It follows that the number of graphs in $\mathcal{G}_{n,n,d,\phi}$ of which $G\in \mathcal{G}_{n,n,d,\phi}$ can be a $(\delta,\eps)$-cut sparsifier is bounded by $2^{dn}\cdot{{n \choose 2} \choose O(dn)}=2^{O(dn\log n)}$. 
	 Fixing the right constant for $d$, it follows that there is a set $N$ of  at least $\frac{2^{\Omega(dn\log n)}}{2^{O(dn\log n)}}=2^{\Omega(dn\log n)}=2^{\Omega(\frac{n\log n}{(\delta+\eps)^{2}})}$ different graphs in $\mathcal{G}_{n,n,d,\phi}$ which are not $(\eps,\delta)$-cut sparsifiers of each other. Taking $\delta'=\frac\delta3$. $\eps'=\frac\eps3$, it holds that no $(\delta,\eps)$-cut sketching scheme $\sk$ can satisfy two graphs from $N$ simultaneously. The theorem now follows.	 
\end{proof}
\section{Finding balanced sparse cuts on sparsifiers} \label{sec:balancedSparseCut}
The gist of this section is that a balanced sparse cut procedure can be run on a $(\delta, \eps)$-cut sparsifier instead of the original graph with almost the same results. To ease the discussion we introduce the following definition.

\begin{definition}[$(\alpha, b)$-balanced sparse cut algorithm] \label{def:BalancedSparseCut}
	Let $\alpha \geq 1$, $b \leq 1$ be some parameters. An algorithm is an $(\alpha, b)$-balanced sparest cut algorithm if, given a parameter $\phi \in [0, 1]$ and access to some information about the graph $G$, it either
	\begin{itemize}
		\item declares that $G$ is a $\phi$-expander\\\\
		or
		\item returns a cut $S$ such that $\Phi_{G}(S)\le \alpha \cdot \phi$, and for
		every cut $S'$ in $G$ where $\Phi_{G}(S')\le\phi$, $\bal_{G}(S) \ge b \cdot \bal_{G}(S')$.
	\end{itemize}
\end{definition}

We adopt two different algorithms for finding a balanced sparse cut and analyze their output when applied to a $(\delta,\eps)$-cut sparsifier. 
The first is a brute force, exhaustive search algorithm by Spielman and Teng (Lemma 7.2 in \cite{ST11}). 
It checks all possible cuts in the sparsifier and among them finds the most balanced one of sparsity at most $\phi$. When applied to a sparsifier, it gives an exponential time algorithm with arbitrarily small sparsity approximation error $\alpha$ and optimal cut balance $b$.

The second adopted algorithm is a balanced sparse cut algorithm of Saranurak and Wang \cite{SW19}. The running time is near linear, but the error $\alpha$ is polylogarithmic in the size of the graph.

Because $b$ is constant and $\alpha$ is $O(\poly\log n)$ in both algorithms that we use, we will assume that $b = O(1)$ and $\alpha = \wt{O}(1)$ fot the rest of the paper.

We begin with the exhaustive search algorithm \cite{ST11}. Note that we measure volume here with respect to the graph $G$. One can measure the volume in the sparsifier, but that would worsen both $\alpha$ and $b$ by a $1 + O(\delta)$ factor.

\begin{lemma}[Exponential time balanced sparse cut]
	\label{lem:BalancedSparseCutExp}
	Consider a graph $G$, and parameters $\phi \in(0,1)$ and $\delta \in (0, \frac{1}{16})$. Then there exists a $(1 + 5 \delta, 1)$-balanced sparse cut algorithm taking as input $\phi$, a $(\delta, \delta \cdot \phi)$-cut sparsifier $G_{\delta, \delta \cdot \phi}$ of $G$, as well as all the degrees of the vertices in $G$.

	This algorithm runs in time $O(2^n m)$ and space $O(n + m)$, where $m = |E(G_{\delta, \delta \cdot \phi})|$.
\end{lemma}
\begin{proof}
	\sloppy Let $\psi := \delta \cdot \phi$. The algorithm is to find the most balanced cut of sparsity at most $ (1 + 2\delta) \phi $ in $G_{\delta, \psi}$, while using the degrees of vertices in graph $G$ to compute the volume.
	Formally, let $S \subset V$ be an arbitrary cut such that $\vol_G(S) \leq \vol_G(\overline{S})$. We define $\Phi'(S) := \frac{w_{G_{\delta, \psi}}(S, \overline{S})}{\vol_G(S)}$ to be the estimation of sparsity of $S$.
	Let 
	$$\mathcal{S}=\left\{S\subset V\mid \vol(S)\le\frac12\vol(G)\mbox{ and }\Phi'(S)\le(1+2\delta)\phi\right\}.
	$$ If $\mathcal{S}=\emptyset$ is empty, we declare that $G$ is $\phi$-expander. Otherwise we return a cut $S\in\mathcal{S}$ that maximizes the volume $\vol(S)$. This finishes the description of the algorithm.
	
	As $G_{\delta,\psi}$ is a $(\delta,\psi)$-cut sparsifier, for every subset $S$,
	\[
		(1 - \delta)w_G(S, \overline{S}) - \psi \vol(S) \leq w_{G_{\delta, \psi}}(S, \overline{S}) \leq (1 + \delta)w_G(S, \overline{S}) + \psi \vol(S).
	\]
	Dividing the above by $\vol(S)$, we get
	\begin{equation}
		(1-\delta)\Phi(S)-\psi\leq\Phi'(S)\leq(1+\delta)\Phi(S)+\psi~.\label{eq:BalancedCutLemmaST}
	\end{equation}

	First, suppose that no cuts of sparsity less than $(1 + 2\delta) \phi$ have been found (i.e. ${\cal S}=\emptyset$). Then 
	for every cut $S$:
	\[
	(1+2\delta)\phi\le\Phi'(S)\stackrel{(\ref{eq:BalancedCutLemmaST})}{\le}(1+\delta)\Phi(S)+\psi=(1+\delta)\Phi(S)+\delta\phi~,
	\]
	implying that $\Phi(S) \geq \phi$, i.e. $G$ is a $\phi$-expander.

	Next, suppose that ${\cal S}\ne\emptyset$, and let $S\in{\cal S}$ be the cut of maximum volume in ${\cal S}$ returned by the algorithm. It holds that
	\[
	(1-\delta)\Phi(S)-\delta\phi\stackrel{(\ref{eq:BalancedCutLemmaST})}{\le}\Phi'(S)\le(1+2\delta)\phi~.
	\]
	Therefore, $\Phi(S) \leq \frac{1 + 3 \delta}{1 - \delta} \phi < (1 + 5 \delta) \phi$ (recall that $\delta<\frac{1}{16}$).
	Finally, consider a cut $S'$ of sparsity at most $\phi$ in $G$. Then 
	\[
		\Phi'(S') \stackrel{(\ref{eq:BalancedCutLemmaST})}{\le} (1 + \delta) \Phi(S') + \delta \phi \leq (1 + 2\delta) \phi~.
	\]
	In particular, $S'\in{\cal S}$, implying that $\bal_G(S') \leq \bal_G(S)$, since $S$ is the most balanced cut returned by the algorithm. 
\end{proof}

Next, we present the polynomial time algorithm. First, we state the balanced sparse cut algorithm of Saranurak and Wang that we will use as a black box.
\begin{lemma}[Corollary of Theorem 1.2 of \cite{SW19}, see Section 4.2, Balanced low-conductance cut]
	\label{lem:BalancedSparseCutBase}
	There exist universal constants $B_1,b_2$ and a $(B_1 \log^3 m, b_2)$-balanced sparse cut algorithm that takes as an input a graph $G=(V, E, w)$ and a parameter $\phi \in (0, 1)$.

	The algorithm runs in time $\wt{O}(m/\phi)$ using space $\wt{O}(m)$ where $m = |E|$.
\end{lemma}

We run this algorithm on the sparsifier.
\begin{lemma}
	\label{lem:BalancedSparseCut}
	Consider a graph $G$, and parameters $\phi \in(0,1)$ and $\delta \in (0, \frac{1}{16})$. There exists a $((1 + 6 \delta )B_1 \log^3 m, (1 - 8 \delta)b_2)$-balanced sparse cut algorithm taking as input $\phi$ and a $(\delta, \delta \cdot \phi)$-cut sparsifier $G_{\delta, \delta \cdot \phi}$ of $G$, where $m = |E(G_{\delta, \delta \cdot \phi})|$.

	The algorithm runs in time $\wt{O}(m/\phi)$ using space $\wt{O}(m)$.
\end{lemma}

\begin{proof}
	Let $\psi := \delta \cdot \phi$. The algorithm uses \Cref{lem:BalancedSparseCutBase} with parameter $\phi'=(1 + 5\delta) \phi $ on the sparsifier $G_{\delta, \psi}$, and returns the same answer.
	Denote $G' = G_{\delta, \psi}$, and $w'$ the weight function of $G'$. Let $S \subset V$ be an arbitrary cut such that $\vol_G(S) \leq \vol_G(\overline{S})$.
	The volume of $S$ equal to the sum of volumes of $S$ vertices (which equals to the singleton cuts), we have
	\[
		\vol_{G'}(S) \leq (1 + \delta)\vol_G(S) + \psi \vol_G(S) \leq (1 + 2 \delta) \vol_G(S)
	\]
	Similarly, $\vol_{G'}(S) \geq (1 - 2\delta) \vol_G(S)$. On the other hand,
	\[
		(1 - \delta)w_G(S, \overline{S}) - \psi \vol_G(S) \leq w_{G'}(S, \overline{S}) \leq (1 + \delta)w_G(S, \overline{S}) + \psi \vol_G(S).
	\]

	Hence, by definition of $\Phi_G(S),\Phi_{G'}(S)$, 
	\begin{align}
		\Phi_{G'}(S) & =\frac{w_{G'}(S,\overline{S})}{\min\{\vol_{G'}(S),\vol_{G'}(\overline{S})\}}\le\frac{(1+\delta)\cdot w_{G}(S,\overline{S})+\psi\cdot\vol_{G}(S)}{(1-2\delta)\cdot\min\{\vol_{G}(S),\vol_{G}(\overline{S})\}}=\frac{1+\delta}{1-2\delta}\cdot\Phi_{G}(S)+\frac{\psi}{1-2\delta}~.\nonumber\\
		\Phi_{G'}(S) & \ge\frac{(1-\delta)\cdot w_{G}(S,\overline{S})-\psi\cdot\vol_{G}(S)}{(1+2\delta)\cdot\min\{\vol_{G}(S),\vol_{G}(\overline{S})\}}=\frac{1-\delta}{1+2\delta}\cdot\Phi_{G}(S)-\frac{\psi}{1+2\delta}~.\label{eq:volume}
	\end{align}
	

	Suppose first that \Cref{lem:BalancedSparseCutBase} returns that $G'$ is a $(1 + 5 \delta)\phi$-expander.
	For every cut $S \subset V$ such that $\vol_G(S) \leq \vol_G(\overline{S})$, as $\delta < \frac{1}{16}$, by \cref{eq:volume} it holds that 
	\[
	\Phi_{G}(S)\ge\frac{1-2\delta}{1+\delta}\cdot\Phi_{G'}(S)-\frac{\psi}{1+\delta}\ge\frac{1-2\delta}{1+\delta}\cdot(1+5\delta)\cdot\phi-\frac{\delta\cdot\phi}{1+\delta}=\frac{1+2\delta-10\delta^{2}}{1+\delta}\cdot\phi>\phi~.
	\]
	Thus $G$ is indeed a $\phi$-expander, and our answer was correct.

	Otherwise, denote by $S$ the returned cut. Let $\beta = B_1 \log^3 m \geq 1$. Then $\Phi_{G'}(S) \leq \beta \phi$.
	Hence, by \cref{eq:volume} and the fact $\delta < 1/4$,
	\[
		\Phi_{G}(S) \leq \frac{(1 + 2 \delta)\beta \phi + \delta \phi}{1 - \delta} \leq (1 + 4 \delta ) \beta \phi + 2 \delta \phi \leq (1 + 6 \delta) \beta \phi~.
	\]

	Finally, for any cut $S'$, by the inequality on the volumes,
	\[
	\frac{1-2\delta}{1+2\delta}\bal_{G}(S')\leq\bal_{G'}(S')=\frac{\min\left\{ \text{Vol}_{G'}(S),\text{Vol}_{G'}(\bar{S})\right\} }{\text{Vol}_{G'}(V)}\leq\frac{1+2\delta}{1-2\delta}\bal_{G}(S')~.
	\]	
	For every cut $S'$ such that $\Phi_{G}(S') \leq \phi$, it holds that 
	$\Phi_{G'}(S)\le\frac{1+\delta}{1-2\delta}\cdot\Phi_{G}(S)+\frac{\psi}{1-2\delta}\le\frac{1+2\delta}{1-2\delta}\cdot\phi\le(1+5\delta)\cdot\phi$. Thus by \Cref{lem:BalancedSparseCutBase} $\bal_{G'}(S) \geq b_2 \bal_{G'}(S')$, implying
	\[
	\bal_{G}(S)\geq\frac{1-2\delta}{1+2\delta}\bal_{G'}(S)\geq\frac{1-2\delta}{1+2\delta}b_{2}\bal_{G'}(S')\geq\left(\frac{1-2\delta}{1+2\delta}\right)^{2}b_{2}\bal_{G}(S)\ge(1-8\delta)b_{2}\bal_{G}(S)
	\]

	Since the algorithm essentially only calls the algorithm from \Cref{lem:BalancedSparseCutBase} once, its running time and memory usage are the same.
\end{proof}

\section{Expander Decomposition via Power Cut Sparsifiers} \label{sec:expanderDecompositionProof}
In this section we show how to use a collection of power cut sparsifiers and a balanced sparse cut lemma to construct an expander decomposition. We present the abstracted version of our algorithm, then we show the final guarantees when using either of the balanced sparse cut lemmas. The exponential time algorithm will allow us to construct expander decomposition with optimal up to a constant factor parameters, while the polynomial time algorithm will have much weaker guarantees.

Our algorithm is taken from \cite{chang2019improved}, but we replace their balanced sparse cut procedure with ours.  The algorithm consists of two stages.
First, \Cref{alg:streaming_ed} applies a balanced sparse cut lemma with sparsity parameter $\phi_0 \simeq \eps$. If the graph is declared to be an expander, the algorithm halts.
Otherwise, if the returned cut $S$ is balanced enough, \Cref{alg:streaming_ed} recurses on both $S$ and $\bar{S}$, and returns the union of the two decompositions.
In the case where the cut $S$ is unbalanced, the algorithm enters a second stage.

If the algorithm finds an unbalanced cut, it is an evidence that the current cluster is almost an expander. The purpose of the second stage, \Cref{alg:decomp_unbal}, is to find a subgraph of the input cluster that is an expander, and partition the rest into singletons, obtaining an expander decomposition. 
\Cref{alg:decomp_unbal} reduces the balancedness requirement to $\frac b2\cdot\eps\vol_G(C)$, and makes cuts of sparsity $\phi_1 = \phi_0 / \alpha$ while they exist, and are balanced enough.
The number of times the algorithm can find a balanced sparse cut is strictly bounded (by the properties of the balanced sparse cut lemma leading us here).
If no cut can be found, the current cluster is an expander, and the algorithm halts. Otherwise, at some step an unbalanced cut is returned. In this case we reduce the sparsity and balancedness requirements, and continue in the same manner. We are guaranteed that after certain number of such reductions, the algorithm will declare the cluster to be expander and halt.
The goal of this process and reductions is to bound the number of executions of the balanced sparse cut lemma, as each such execution requires a fresh cut sparsifier, translating to a larger memory requirement.

Both \Cref{alg:streaming_ed,alg:decomp_unbal} assume that the sparsifiers of $G$ that they use are already constructed.

\begin{figure}[!h]
 \fbox{\begin{minipage}{43.5em}
\begin{multicols}{2}
\begin{flushleft}
	{\small
\begin{itemize}
	\item $\alpha, b$ --- balanced sparse cut parameters (\Cref{def:BalancedSparseCut})
   	\item $\eps$ --- desired fraction of inter-cluster edges.
	\item $k \in \mathbb{N}$	--- quality parameter.
	\item $\mathrm{O}$ --- upper bound on $\vol(G)$.
 	\item $j \leq k$ usually denotes a number of a particular iteration of the outer loop of \Cref{alg:decomp_unbal}, $k' \leq k$ --- the number of the last iteration.
	\item $\delta \leq 1/16$ --- multiplicative parameter of power cut sparsifier.
	\item $\phi_0=\frac{\epsilon}{2\log(\vol(G)) \alpha}$, 
	$\forall j \in \mathbb{N},~\phi_{j} = \phi_{j - 1} / \alpha$ --- intermediate sparsity values parameters of \Cref{alg:decomp_unbal}.
	\item $\forall j \in \mathbb{N} \cup \{0\},~\psi_j=\delta \cdot \phi_j$ --- additive parameters of power cut sparsifiers.
	\item $G_{\psi_0,\delta}^{1},\dots,G_{\psi_0,\delta}^{\wt{O}(\frac{1}{\epsilon})}$ --- sparsifiers for \Cref{alg:streaming_ed}.
	\item $\forall j \in [k+1]$, $G_{\psi_j,\delta}^{1},\dots,G_{\psi_j,\delta}^{O(\mathrm{T})}$ --- sparsifiers for \Cref{alg:decomp_unbal}.
	\item $T_j(C)$ --- running time of the algorithm from \Cref{def:BalancedSparseCut} on $G_{\psi_j, \delta}\{C\}$ and $\phi=\phi_j$.
	\item $\tau = (\epsilon \cdot \vol_{G}(C))^{1/k}$ --- volume decrease step in \Cref{alg:decomp_unbal}.
 	\item $\mathrm{T} = \lceil (\eps \cdot \vol_G(G))^{1/k} \rceil$ --- upper bound on $\tau$.
  	\item $m_1 = \epsilon \cdot \vol_{G}(C)$, $\forall j \in \mathbb{N}$, $m_{j + 1} = m_j / \tau$ ---
	  volume thresholds in \Cref{alg:decomp_unbal}.
\end{itemize}}
\end{flushleft}
\end{multicols}	
 \end{minipage}}	
\caption{Summary of the parameters and variables used in \Cref{alg:streaming_ed} and \Cref{alg:decomp_unbal}}
\label{fig:parameters}
\end{figure}

\begin{algorithm}[!h]
	\caption{\texttt{Low-depth-expander-decomposition}$\left(C,\epsilon,k, i\right)$}\label{alg:streaming_ed}
	\DontPrintSemicolon
	\SetKwInOut{Input}{input}\SetKwInOut{Output}{output}
	\Input{cluster $C\subseteq V$, parameter bounding the number of outer-cluster edges $\eps$, quality parameter $k$, counter for the depth of the recursion $i$.}
	\Output{expander decomposition of $C$.}
	\SetKwFunction{Lded}{Low-depth-expander-decomposition}
	\SetKwFunction{Ubcd}{Unbalanced-cluster-decomposition}
	Use an $(\alpha,b)$-balanced sparse cut algorithm (\Cref{def:BalancedSparseCut}) with the sparsifier $G_{\psi_0,\delta}^i\{C\}$ and parameter $\phi = \phi_0$
	 \tcp*{$G_{\psi_0,\delta}^{i}$ is global knowledge, recall $\psi_0=\delta\cdot \phi_0$}
	\If{it declares that $C$ is $\phi_0$-expander}
	{\Return{$C$}}
	\Else{
		Let $S$ be the returned $\alpha \cdot \phi_0$-sparse cut\;
		\If{$\text{Vol}_{G}(S)\ge\frac{\eps\cdot b}{4}\cdot\text{Vol}_{G}(C)$}{
			\Return  \Lded{$S,\epsilon,k, i+1$} $\cup$
			\Lded{$C\setminus S,\epsilon, k,i+1$} \label{line:returnBalanced}\;
		}
		\Else{
			\Return \Ubcd{$C, \epsilon, k$}\;
		}
	}
\end{algorithm}

\begin{algorithm}[!h]
	\caption{\texttt{Unbalanced-cluster-decomposition}$(C_0, \epsilon, k)$}\label{alg:decomp_unbal}
	\DontPrintSemicolon
	\SetKwInOut{Input}{input}\SetKwInOut{Output}{output}
	\Input{Cluster $C_0\subseteq V$ which does not contain a balanced sparse cut, parameter bounding the number of outer-cluster edges $\eps$, quality parameter $k$.}
	\Output{Expander decomposition of $C_0$.}
	$\phi_1 \gets \frac{\epsilon}{2\log(\vol(G))\alpha^2}$\;
	$C \gets C_0$, $\tau \gets  (\epsilon \cdot \vol_{G}(C))^{1/k}$, $m_1 \gets \epsilon \cdot \vol_{G}(C)$\;
	\For{$j=1$ to $\infty$}{
		\For{$h=1$ to $\infty$}{
			$G_{\psi_j,\delta}^{C}\leftarrow G_{\psi_j,\delta}^{h }\{C\}$\tcp*{$G_{\psi_j,\delta}^{h}$ is global knowledge}
			Run an $(\alpha,b)$-balanced sparse cut algorithm (\Cref{def:BalancedSparseCut}) on the sparsifier $G_{\psi_j,\delta}^{C}$ and parameter $\phi = \phi_j$\;
			\If{it declares that $C$ is $\phi_j$-expander}{
				\Return $C$, $\left\{ u\right\} _{u\in {C_0 \setminus C}}$ \label{line:ubcd_return}
				\tcp*{The only return point of this algorithm}
			}
			\Else{
				Let $S$ be the returned $\alpha \cdot  \phi_j$-sparse cut \label{line:ubcd_cut}\;
				\If{$\text{Vol}_{G}(S)\ge \frac{b}{2} m_j$}{
					$C \gets C \setminus S$
				}
				\Else {
					\label{line:ubcd_break}
					\textbf{break} inner loop, \textbf{goto} \cref{line:updateJ}\;
				}
			}
		}
		$m_{j + 1} \gets m_j / \tau$, $\phi_{j + 1} \gets \phi_j / \alpha $\label{line:updateJ}\;
	}
\end{algorithm}

\subsection{Analyses of the algorithms} 


We begin by discussing the input parameters.
$\eps$ is required to be in $(0,\frac14)$. $k \in \mathbb{N}$ is a parameter regulating the tradeoff between the required space (number of sparsifiers we will be using) and the quality of the expander decomposition. Let $\delta \leq 1/16$ be the parameter of power cut sparsifier, $\alpha, b$ be the balanced sparse cut lemma parameters.
Fix $\phi_0=\frac{\epsilon}{2\log(\vol(G)) \alpha}$ and $\phi_{j} = \phi_{j - 1} / \alpha$ for $j \in \mathbb{N}$. Similarly, let $\psi_j=\delta \cdot \phi_j$ for $j \in \mathbb{N} \cup \{0\}$.
Let $\mathrm{T} = \lceil (\eps \cdot \vol_G(G))^{1/k} \rceil$. Parameters $\alpha$ and $b$ will depend on the balanced sparse cut lemma used, but in both cases $\alpha=\wt{O}(1)$, $b = O(1)$.

In the execution of \Cref{alg:streaming_ed,alg:decomp_unbal}, we assume that we have access to $\wt{O}(\frac{1}{\epsilon})$ independent $(\psi_0,\delta)$-cut sparsifiers $G_{\psi_0,\delta}^{1},G_{\psi_0,\delta}^{2},\dots$, and for each $j \in [k+1]$ $O(\mathrm{T})$ independent $(\psi_j, \delta)$-cut sparsifiers $G_{\psi_j,\delta}^{1}, G_{\psi_j,\delta}^{2}, \ldots$.
Even though the algorithm is recursive, we only need a number of sparsifiers proportional to the depth of the recursion, since if we were to consider a set of the subgraphs from the same depth that we run our algorithm on, this set would be a subset of some partition of $G$. Hence, by the property of power cut sparsifiers we can use a single sparsifier for all the sets in a single depth level of the recursion in \Cref{alg:streaming_ed}.
The same kind of argument works for \Cref{alg:decomp_unbal}.
Note that \Cref{alg:decomp_unbal} uses a different set of sparsifiers from that of \Cref{alg:streaming_ed}. Also note that $\mathrm{T} \leq \vol(G)^{1/k} \leq n^{2/k}$.

The expander decomposition will be obtained by calling \texttt{Low-depth-expander-decomposition}$\left(V,\epsilon, k, 1\right)$.


We start by analyzing \Cref{alg:decomp_unbal}. In particular, first we bound the total volume cut by it. We make use of a simple composition property of sparse cuts.

\begin{lemma}
	\label{lem:SparseCutSum}
	Let $G=(V, E)$ be an arbitrary graph, and let $\phi \in [0, 1]$. Consider a pair of cuts: $S_1$ which is $\phi$-sparse in $G$, and $S_2$ which is   $\phi$-sparse in $G\setminus S_1$. 
	Then, if $\vol_G(S_1 \cup S_2) \leq \frac12\vol_G(V)$, the cut $S_1 \cup S_2$ is $\phi$-sparse in $G$.
\end{lemma}
\begin{proof}
	From conditions on cuts $S_1, S_2$ the following  is immediate:
	\[
		| E(S_1, V \setminus S_1)| \leq \phi\cdot \vol_G(S_1)
	\]
	\[
		| E(S_2,( V \setminus S_1) \setminus S_2)| \leq \phi \cdot\vol_G(S_2)
	\]
	Therefore,
	\begin{align*}
		|E(S_{1}\cup S_{2},V\setminus(S_{1}\cup S_{2}))| & \leq|E(S_{1},V\setminus S_{1})|+|E(S_{2},V\setminus(S_{1}\cup S_{2}))|\\
		& \leq\phi\cdot\vol_G(S_{1})+\phi\cdot\vol_G(S_{2})~~=~~\phi\cdot\vol_G(S_{1}\cup S_{2})
	\end{align*}
	Hence, $S_1 \cup S_2$ also has sparsity at most $\phi$.
\end{proof}
For a fixed invocation of \Cref{alg:decomp_unbal}, we denote by $k'$ the number of the outer loop iteration at the time where the algorithm returns (i.e. the value of the index $j$ in \cref{line:ubcd_return}). Let $C_j$ be the value of the variable $C$ at the end of iteration $j$ of the outer loop. In particular, $C_{k'}$ is the returned non-singleton cluster in \cref{line:ubcd_return}.
For $j<k'$, the balanced sparse cut algorithm did not declare $C_j$ to be an expander, but also was not able to find a balanced sparse cut.
In \Cref{lem:ubcd_iteration_volume_guarantee} we argue that the volume of all future cuts made by the algorithm (the volume of $C_j\setminus C_{k'}$) is very small.

\begin{lemma} \label{lem:ubcd_iteration_volume_guarantee}
	For $0 \leq q \leq k'$ consider $C_q$ and suppose that in $C_q$ there is no cut $S$ of sparsity at most $\phi_q$ and volume satisfying $Z\le\vol_G(S) \le \frac12\vol_G(C_q)$ for some parameter $Z\leq \frac14\vol_G(C_q)$.\\
 Then $\vol_G(C_q \setminus C_{k'}) < Z$.
\end{lemma}
%
\begin{proof}
	Suppose that in the iterations $\{q+1,\dots,k'\}$ the algorithm iteratively created the cuts $S_1,\dots,S_w$. Note that $C_{k'}=C_q\setminus\cup_{i=1}^{w}S_i$.
	Notice that all these cuts $S_i$ are found after the iteration $q$, and hence have sparsity at most $\alpha \cdot \phi_{q + 1} = \phi_q $.
	Further, note that every cut $S_i$ which is $\phi_q$-sparse in $C_q\setminus\cup_{i'=1}^{i-1}S_{i'}$, is also $\phi_q$-sparse in $C_q$. 
	Therefore, by the assumption of the lemma $\vol(S_i) < Z$.
	
	Assume by contradiction that $\vol_G(\cup_{i=1}^{w}S_i)=\vol_G(C_q \setminus C_{k'}) \geq Z$. Let $w'\in[1,w]$ be the minimal index such that $\vol_G(\cup_{i=1}^{w'}S_i)\geq Z$. By minimality, and since $\vol(S_{w'}) < Z$, we have that $$\vol_G(\cup_{i=1}^{w'}S_i)=\vol_G(\cup_{i=1}^{w'-1}S_i)+\vol_G(S_{w'})\leq Z+Z=\frac12\vol_G(C_q)~.$$
	
	As all the cuts $S_1,\dots,S_{w'}$ are $\phi_q$-sparse, using \Cref{lem:SparseCutSum} inductively we conclude that $\cup_{i=1}^{w'}S_i$ is $\phi_q$-sparse cut in $C_q$, a contradiction with the assumption of the lemma.
\end{proof}
This finally allows us to bound the total volume of cut edges by \Cref{alg:decomp_unbal}.
\begin{corollary} \label{cor:ucd_volume_guarantee} 
	The total volume of inter-cluster edges in the partition returned in \cref{line:ubcd_return} by \Cref{alg:decomp_unbal}, i.e. $\vol_G(C_0 \setminus C_{k'})$, is bounded by $\frac\eps2\cdot \vol_G(C_0)$.
\end{corollary}
\begin{proof}
	\Cref{alg:decomp_unbal} is only called in \Cref{alg:streaming_ed} when the balanced sparse cut algorithm finds a $\alpha\cdot\phi_0$-sparse cut $S$ such that $\vol_G(S) \leq \frac{\eps \cdot b}{4}\cdot \vol_G(C)$, where $C = C_0$. 
	As we are using $(\alpha,b)$-balanced sparest cut algorithm, 
	  any cut of sparsity at most $\phi_0$ in $C_0$ has volume at most $\frac{\eps \cdot b}{4(1 - 8\delta)b} \cdot\vol_G(C) \leq \frac\eps2\cdot \vol_G(C)$ (as $\delta<\frac{1}{16}$). Therefore, since $\eps < 1/4$, by \Cref{lem:ubcd_iteration_volume_guarantee}, $\frac\eps2 \cdot\vol(C_0 \setminus C_{k'}) \leq \frac{\eps}{2} \vol_G(C)$, which upper bounds the volume of inter-cluster edges.
\end{proof}

Another application of \Cref{lem:ubcd_iteration_volume_guarantee} is bounding the number of inner loop iterations of \Cref{alg:decomp_unbal}. This is necessary to bound the number of sparsifier that the algorithm uses. We also show the correctness, i.e. that the algorithm returns an expander decomposition.

\begin{lemma}\label{lem:numberOfIterations}
	The following statements regarding \Cref{alg:decomp_unbal} hold:
	\begin{itemize}
		\item \Cref{alg:decomp_unbal} returns a $\phi_{k + 1}$-expander decomposition in \cref{line:ubcd_return} before the iteration $k+1$ of the outer loop ends.
		\item For each iteration of the outer loop, the inner loop halts (in \cref{line:ubcd_return}) or breaks (in \cref{line:updateJ}) before the iteration $ \lfloor \frac{\tau}{b}  \rfloor + 1 = O(\mathrm{T})$ ends.  		
	\end{itemize}
\end{lemma}
\begin{proof}
	At the iteration $k + 1$ (if it occurs), $m_{k + 1}=m_1\cdot\tau^{-k}= 1$. Thus $\vol_G(S) < \frac{b}{2} m_{k + 1}\le\frac12$ implies that $\vol_G(S) = 0$ and the cut is empty.
	But in this case, any balanced sparest cut algorithm  (\Cref{def:BalancedSparseCut}) would declare $C$ to be a $\phi_{k + 1}$-expander, and the algorithm would return a decomposition into $C$ and singletons (\cref{line:ubcd_return}). Therefore, algorithm always stops before the iteration $k + 1$ of the outer loop terminates, returning a $\phi_{k + 1}$-expander decomposition.
	
	To show the second point, consider an iteration $j \leq k + 1$ of the outer loop.  For $j = 1$, by \Cref{cor:ucd_volume_guarantee}, we have
	$\vol(C_0 \setminus C_1) \leq \vol(C_0 \setminus C_{k'}) \leq \frac{\epsilon}{2} \cdot\vol_G(C_0)$, and since each cut made during the first iteration has volume at least $\frac b2\cdot m_1 = \frac b2\cdot \eps \cdot \vol_G(C)$,
	the algorithm can make at most $1/b<\frac{\tau}{b}$ cuts (and thus inner loop iterations).
	
	For $j > 1$, on the previous $(j-1)$-iteration the algorithm found a cut $S$ of sparsity at most $\alpha\cdot\phi_{j-1}$ such that $\vol_G(S) < \frac{b}{2}m_{j - 1}$ in $C_{j-1}$. 
	Hence by the balancedness guarantee of balanced sparse cut algorithm, 
	any cut of sparsity at most $\phi_{j-1}$ in $C_{j-1}$ has volume at most $\frac{b}{2(1 - 8\delta)b} \cdot m_{j-1} \leq \frac{m_{j-1}}{4}$ (as $\delta<\frac{1}{16}$).
	It follows by \Cref{lem:ubcd_iteration_volume_guarantee} that
	$\vol_G(C_{j-1}\setminus C_j)\le \vol_G(C_{j-1}\setminus C_{k'})<\frac{m_{j-1}}{4}$.
	As each cut in the $j$'th iteration has volume at least $\frac b2 m_j$, the total number of cuts (and thus inner loop iterations) is bounded by $\nicefrac{\frac{m_{j-1}}{4}}{\frac{b}{2}m_{j}}=\frac{1}{2b}\cdot\frac{m_{j-1}}{m_{j}}=\frac{\tau}{2b}$.

\end{proof}

Now we analyze \Cref{alg:streaming_ed}. This algorithm is much more standard, and we use well established tools to do it, so we skip some details.
\begin{lemma}\label{lem:recursionDepth}
	\Cref{alg:streaming_ed} has recursion  depth $\wt{O}(\frac1\eps)$.
\end{lemma}
\begin{proof}
	Since after each iteration the volume decreases by a factor at least $1 - \frac{\eps\cdot b}{4}$ and the volume is polynomial in $n$, the depth of the recursion is at most $\log_{\frac{1}{1-\frac{\eps\cdot b}{4}}}\left(\vol(G)\right)=O(\frac{\log n}{\epsilon\cdot b})=\wt{O}(\frac1\eps)$.
\end{proof}
The goal of next two lemmas is to show that the combination of two algorithms returns $(\eps, \phi_{k + 1})$-expander decomposition
\begin{lemma}\label{lem:ExpansionParameter}
	Each final cluster is a $\phi_{k+1}$-expander, where $\phi_{k + 1} = \frac{\epsilon}{2\log(\vol(G))} \alpha^{-k - 2}$.
\end{lemma}
\begin{proof}
	Any final cluster is either a singleton, or was created in
	either \Cref{alg:streaming_ed} or \Cref{alg:decomp_unbal} after applying the algorithm from \Cref{def:BalancedSparseCut}. Because the minimum conductance value that we run this algorithm with is $\phi_{k + 1}$, every such cluster is at least a $\phi_{k + 1}=\phi_1\cdot\alpha^{-k}=\frac{\epsilon}{2\log(\vol(G))\alpha^{k+2}}$ expander.
\end{proof}

To bound the number of edges cut in \Cref{alg:streaming_ed} we use a standard charging argument.
\begin{lemma}\label{lem:InterClusterVolume}
	The volume of inter cluster edges is at most $\epsilon\cdot\vol(G)$.
\end{lemma}
\begin{proof}
	We present a charging scheme to bound the overall volume of inter cluster edges.
	In \Cref{alg:streaming_ed}, cuts occur only in \cref{line:returnBalanced} to clusters $S$ and $\bar{S}$, where $\vol_G(S)\le\vol_G(\bar{S})$.
	When that happens, each vertex in $v\in S$ will be charged $\alpha \phi_0 \cdot\deg(v)$.
	The total charge is $\sum_{v\in S}\alpha\phi_{0}\cdot\deg(v)=\alpha\phi_{0}\cdot\vol_{G}(S)$, which is enough to pay for all the edges (as the cut is $\alpha\cdot\phi_0$-sparse).
	Each time a vertex $v$ is charged, the volume of the cluster it belongs to at least halves.  Therefore each vertex could be charged
	at most $\log \vol(G)$ times. In total, we charge all vertices for cuts occurring in \Cref{alg:streaming_ed} at most $\log(\vol(G))\cdot\sum_{v\in V}\alpha\phi_{0}\cdot\deg(v)=\alpha\phi_{0}\cdot\vol(G)\cdot\log(\vol(G))=\frac\eps2\cdot\vol(G)$.
	
	By \Cref{cor:ucd_volume_guarantee}, when executing \Cref{alg:decomp_unbal} on a cluster $C_0$, the total volume cut is bounded by $\frac\eps2\cdot\vol(C_0)$. Since each vertex could belong only to a single cluster on which \Cref{alg:decomp_unbal} is executed, the total volume of the edges cut due to  \Cref{alg:decomp_unbal} is at most $\frac\eps2\cdot\vol(G)$.
	We conclude that the total volume of intercluster edges is at most $\eps\cdot\vol(G)$.
\end{proof}


Our algorithm requires $\wt{O}(\eps^{-1})$ 
samples from the $(\delta,\psi_0,n^{-\Czeta})$-power cut sparsifier distribution, for some $\Czeta>0$, and for $j\in[1,k]$, $O(T)$ samples from the  $(\delta,\psi_j,n^{-\Czeta})$-power cut sparsifier distribution. We also assume that the sparsifiers were drawn independently.

\Cref{lem:recursionDepth} together with \Cref{lem:numberOfIterations} show that this number of sparsifiers is indeed enough to run the entire algorithm.	For \Cref{alg:streaming_ed} we use a single sparsifier for each recursion level. Specifically, if the set of clusters an the $i$'th level of the recursion is $C^i_1,C^i_2,\dots$,then the sparsifiers that will be used are $G^i_{\psi_0,\delta}[C^i_1],G^i_{\psi_0,\delta}[C^i_2],\dots$. The required number of sparsifiers is bounded by the recursion depth.
For \Cref{alg:decomp_unbal}, we will use the exact same set of sparsifiers for every cluster sent to the algorithm. By \Cref{lem:numberOfIterations} we indeed need only $k+1$ types of sparsifiers, $O(T)$ of each type.
%
%
Finally, we analyze the space and time consumption of the algorithm .
The space required to store all the sparsifiers is
\begin{align}
	\wt{O}(\eps^{-1})\cdot\left|G_{\psi_{0},\delta}^{1}\right|+\sum_{j=1}^{k+1}\sum_{q=1}^{O(T)}\left|G_{\psi_{j},\delta}^{1}\right| & =\wt{O}\left(\eps^{-1}\cdot\frac{n\Czeta}{\psi_{0}\cdot\delta}\right)+O\left(T\cdot\frac{n\Czeta\cdot\log^{2}n}{\delta}\cdot\sum_{j=1}^{k+1}\frac{1}{\psi_{j}}\right)\nonumber\\
	& =\wt{O}(n)\cdot\frac{\Czeta}{\delta^{2}}\cdot\left(\frac{1}{\eps\cdot\phi_{0}}+T\cdot\sum_{j=1}^{k+1}\frac{1}{\phi_{j}}\right)\nonumber\\
	& =\wt{O}(n)\cdot\frac{\Czeta}{\delta^{2}\cdot\phi_{0}}\cdot\left(\frac{1}{\eps}+T\cdot\sum_{j=1}^{k+1}\alpha^{j}\right)\nonumber\\
	& =\wt{O}(n)\cdot\frac{\Czeta\cdot\alpha}{\delta^{2}\cdot\eps}\cdot\left(\frac{1}{\eps}+(\eps\cdot\vol_{G}(G))^{1/k}\cdot\frac{\alpha^{k+2}-\alpha}{\alpha-1}\right)~.\label{eq:SpaceBound}
\end{align}

The only other component using memory is the algorithm from \Cref{def:BalancedSparseCut}. Because its memory requirement is $\wt{O}(m')$, where $m'$ is the number of edges in the graph that we run it on, and we run it on the sparsifiers, the final memory cost is still the same. 

Next, we analyze running time.
Denote by $T_j(C)$ the time it takes to run the balanced sparse cut algorithm on the sparsifier $G_{\psi_j, \delta}\{C\}$ with sparsity parameter $\phi_j$. Note that for both of the balanced sparse cut algorithms of \Cref{lem:BalancedSparseCutExp,lem:BalancedSparseCut}, it holds that $T_j(C_1) + T_j(C_2) \leq T_j(C_1 + C_2)$ if $C_1$ and $C_2$ are disjoint. 
As each sparsifier is applied only on disjoint clusters, the overall running time of the executions of the balanced sparse cut lemmas is 
\begin{equation}
	\wt{O}(\eps^{-1})\cdot{\cal T}_{0}(V)+\sum_{j=1}^{k+1}O(T)\cdot{\cal T}_{j}(V)\le\max\left\{ \wt{O}\left(\frac{1}{\eps}\right),O(T)\right\} \cdot\sum_{j=0}^{k+1}{\cal T}_{j}(V)~.\label{eq:RunTime}
\end{equation}
All the other operations are fairly straightforward and take $\wt{O}(n\cdot T\cdot k)$.

\begin{lemma} \label{lem:total-sparsif-failure}
	Assuming that each sparsifier was sampled from a distribution with failure probability at most $n^{-\Czeta}$, the probability of failure is at most $\widetilde{O}(\frac1\eps)\cdot n^{-\Czeta +\frac2k}\cdot k$ for a large enough $n$.
\end{lemma}
\begin{proof}
	The only place that the algorithm can fail is the application of sparsifiers to a particular clustering. Since we use each sparsifier only for one partition, it is enough to do a union bound over all used sparsifiers.
	By \Cref{lem:recursionDepth}, the number of sparsifiers needed for the \cref{alg:streaming_ed} is 
	$\wt{O}(1/\eps)$.

	Recall that we use the same sparsifiers for each call to \Cref{alg:decomp_unbal}. The number of used sparsifiers is equal to total number of iteration of the inner loop, which is bounded by \Cref{lem:numberOfIterations} by $O(\mathrm{T}\cdot k)\leq O(k\cdot n^{\frac2k})$, since the graph is unweighted. By union bound, the failure probability is thus bounded by $\widetilde{O}(\frac1\eps)\cdot n^{-\Czeta +\frac2k}\cdot k$.
\end{proof}

\subsection{Proofs of \Cref{thm:main-exp,thm:main-poly}}
We can finally present the proofs of \Cref{thm:main-exp,thm:main-poly} (restated below for convenience). 

Before we start, first we need to remark that while the parameters of the sparsifiers that we use are dependent on $\vol(G)$, we don't know this volume at the start of the stream. Hence we assume that we now some upper bound $\mathrm{O}$ on the volume. For unweighted graphs it is always the case that $\mathrm{O}=n^2$ is sufficient. We use this upper bound to define the lower bound on the quality of the expanders that we use. Namely, $\psi_0' := \delta \frac{\eps}{2 \log(\mathrm{O}) \alpha}$, and $\psi'_i = \psi'_{i - 1} / \alpha$. We construct sparsifiers $G_{\psi'_i, \delta}$ instead of $G_{\psi_i, \delta}$. Since $\psi'_i \leq \psi_i$, they can replace $G_{\psi_i, \delta}$ in the algorithm.

In the theorems we assume that $\log(\mathrm{O}) = O(\log(n))$.

First we use the exponential time \Cref{lem:BalancedSparseCutExp} in \Cref{alg:streaming_ed} to obtain an exponential time algorithm for expander decomposition.

\ExpTime*
\begin{proof}
	Let $\delta = \frac{1}{5 \log n}$, $k = \log n$, $\alpha=1+5\delta$, $b=1$, and $\Czeta$ to be large enough constant. Note that $\alpha^{k}=(1+\frac{1}{\log n})^{\log n}=\Theta(1)$,  $\phi_{k+1}=\phi_{0}\cdot\alpha^{-k-1}=\frac{\epsilon}{2\log(\vol(G))}\cdot\alpha^{-k-2}=\Omega(\frac{\eps}{\log n})$, and $\mathrm{T}=\lceil(\eps\cdot\vol_{G}(G))^{1/k}\rceil=O(1)$.	
	We use \Cref{thm:pcs-alg} during the stream to construct all sparsifiers required by \Cref{alg:decomp_unbal,alg:streaming_ed}.
	Following \Cref{eq:SpaceBound}, the space required to construct (and store) all these sparsifiers is bounded by
	\[
	\wt{O}(n)\cdot\frac{\Czeta\cdot\alpha}{\delta^{2}\cdot\eps}\cdot\left(\frac{1}{\eps}+(\eps\cdot \mathrm{O})^{1/k}\cdot\frac{\alpha^{k+2}-\alpha}{\alpha-1}\right)=\wt{O}(n)\cdot\eps^{-2}~.
	\]
	
	At the end of the stream we execute \Cref{alg:decomp_unbal,alg:streaming_ed} using \Cref{lem:BalancedSparseCutExp} as the balanced sparse cut algorithm. 
	By \Cref{lem:ExpansionParameter,lem:InterClusterVolume,lem:total-sparsif-failure} we obtain a $\left(\eps,\Omega(\frac{\eps}{\log n})\right)$-expander decomposition with high probability. 
	Note that $T_{j}(G) = O(2^n m)$ for all $j$, and thus by \Cref{eq:RunTime},
	the total running time is $\wt{O}(\frac{1}{\eps})\cdot k\cdot O(2^{n}\cdot n^{2})=\wt{O}(\frac{2^{n}}{\eps})$.
\end{proof}

Likewise, we combine polynomial time balanced sparse cut procedure, \Cref{lem:BalancedSparseCut}, with \Cref{alg:streaming_ed} to obtain the following.

\PolyTime*
\begin{proof}
	
		
	Let $\delta = \frac{1}{16}$, $\alpha=(1+6\delta)B_{1}\log^{3}m\le11B_{1}\cdot\log m=\Theta(\log^{3}n)
	$, $b=(1-8\delta)b_{2}=\frac{1}{2}\cdot b_{2}=\Omega(1)$, and $\Czeta$ to be large enough constant. We will also assume that $n$ is large enough (as otherwise the theorem is trivial).
	Note that for $j\ge 1$, $\phi_{j}=\phi_{0}\cdot\alpha^{-j}=\frac{\epsilon}{2\log(\vol(G))}\cdot\alpha^{-j-1}=\Omega(\frac{\eps}{\log^{O(j)}n})$.
	We use \Cref{thm:pcs-alg} during the stream to construct all sparsifiers required by \Cref{alg:decomp_unbal,alg:streaming_ed}.
	Following \Cref{eq:SpaceBound}, the space required to construct (and store) all these sparsifiers is bounded by
	\[
	\wt{O}(n)\cdot\frac{\Czeta\cdot\alpha}{\delta^{2}\cdot\eps}\cdot\left(\frac{1}{\eps}+(\eps\cdot\mathrm{O})^{1/k}\cdot\frac{\alpha^{k+2}-\alpha}{\alpha-1}\right)=\wt{O}(n)\cdot\left(\eps^{-2}+\eps^{\frac{1}{k}-1}\cdot n^{\frac{2}{k}}\cdot\log^{O(k)}n\right)~.
	\]
	
	At the end of the stream we execute \Cref{alg:decomp_unbal,alg:streaming_ed} using \Cref{lem:BalancedSparseCut} (with parameter $\alpha,b$ above) as the balanced sparse cut algorithm.
	By \Cref{lem:ExpansionParameter,lem:InterClusterVolume,lem:total-sparsif-failure} with high probability, we obtain a $(\eps,\phi_{k+1})=(\eps,\Omega(\frac{\eps}{\log^{O(k)}n}))$-expander decomposition

	Note that for every $j$, $T_{j}(G)=O(\frac{n^{2}}{\phi_{j}})=O(\frac{n^{2}}{\eps}\cdot\log^{O(j)}n)$, and thus by \Cref{eq:RunTime},
	the total running time is 	
	\[
	\max\left\{ \wt{O}(\frac{1}{\eps}),O(n^{\frac{2}{k}})\right\} \cdot\sum_{j=0}^{k+1}O(\frac{n^{2}}{\eps}\cdot\log^{O(j)}n)=O(\frac{n^{2}}{\eps}\cdot\log^{O(k)}n\cdot(n^{\frac{2}{k}}+\eps^{-1}))~.
	\]
\end{proof}

%

\section*{Acknowledgments}
Mikhail Makarov is supported by ERC Starting Grant 759471. Michael Kapralov is supported in part by ERC Starting Grant 759471. Arnold Filtser is supported by the ISRAEL SCIENCE FOUNDATION (grant No. 1042/22).

   {\small
  	 \bibliographystyle{alphaurlinit}
  	 \bibliography{bibArnold}
   }

\appendix
\section{Additive Chernoff type inequalities}\label{appendix:AddChernoff}
In this section we will present auxiliary inequalities used in this paper. We will start with a simple corollary of the Chernoff bound.

\begin{theorem}[Theorem 1.1 of \cite{dubhashi2009concentration}]
    \label{thm:chernoff_base}
    Let $X_1, \ldots, X_n$ be independent random variables distributed in $[0, a]$. Let $X = \sum_{i\in [n]}X_i$, $\mu = \E[X]$. Then for $\eps \in (0, 1)$
    \[
        \Pr[X > (1 + \eps) \mu] \leq \exp \left(-\frac{\eps^2 \mu}{3a} \right)
    \]
    \[
        \Pr[X < (1 - \eps) \mu] \leq \exp \left(-\frac{\eps^2 \mu}{2a} \right)
    \]
    \[
        \Pr[|X - \mu| > \eps \mu] \leq 2 \exp \left(-\frac{\eps^2 \mu}{3a} \right)
    \]
\end{theorem}
\begin{corollary}
    \label{cor:add_chernoff_bounded}
    Let $X_1, \ldots, X_n$ be independent random variables distributed in $[0, a]$. Let $X = \sum_{i\in [n]}X_i$, $\mu = \E[X]$. Then for $\eps \in (0, 1)$ and $\alpha \in [0, \mu)$
    \[
        \Pr[|X - \mu| > \eps \mu + \alpha] \leq 2 \exp \left(-\frac{\eps \alpha}{3a} \right)
    \]
\end{corollary}
\begin{proof}
There are two cases:
\begin{itemize}
    \item If $\alpha \leq \eps \mu$, then 
    \[
        \Pr[|X - \mu| > \eps \mu + \alpha] \leq \Pr[|X - \mu| > \eps \mu] \leq 2 \exp \left(-\frac{\eps^2 \mu}{3a} \right) \leq 2 \exp \left(-\frac{\eps \alpha }{3a} \right)
    \]
    \item If $\alpha > \eps \mu$, set $\tilde{\eps} = \alpha / \mu < 1$, $\eps < \tilde{\eps}$ and apply \Cref{thm:chernoff_base} to it:
    \[
        \Pr[|X - \mu| > \eps \mu + \alpha] \leq \Pr[|X - \mu| > \tilde{\eps}\mu] \leq 2 \exp \left(-\frac{\tilde{\eps}^2 \mu}{3a} \right) \leq 2 \exp \left(-\frac{\eps \alpha }{3a} \right)
    \]
\end{itemize}
\end{proof}

The problem with this bound is that it is restricted on the values of $\alpha$. Because of this, we will derive a complementary bound using Bennett's inequality.

\begin{theorem}[Theorem 2.9 of \cite{boucheron2013concentration} and it's corollary]
    \label{thm:bennet_ineq}
    Let $X_1, \ldots, X_n$ be independent random variables with finite variance such that $\E[X_i] = 0$ and $X_i \leq a$ for some $a > 0$ almost surely for all $i \in [n]$. Let $v = \sum_{i \in [n]} X_i^2$. Then for any $t > 0$
    \[
        \Pr\left[\sum_{i \in [n]} X_i \geq t \right] \leq \exp \left(- \frac{v}{a^2}h \left( \frac{at}{v} \right) \right) \leq \exp \left(- \frac{t}{2 a} \ln \left( 1 + \frac{at}{v} \right) \right)
    \]
    where $h(x) = (1 + x)\ln(1 + x) - x$.
\end{theorem}
\begin{proof}
    Here we only show the second inequality.

    It is enough to show that $h(x) \geq \frac{x}{2} \ln(1 + x)$ for $x \geq 0$. Let $f(x) = h(x) - \frac{x}{2} \ln(1 + x)$. Then $f(0) = 0$, $f'(x) = \frac{1}{2}(\ln(1 + x) - x/(1 + x)) = \frac{h(x)}{2(1 + x)} \geq 0$. Therefore, $f(x) \geq 0$ for $x \geq 0$.
\end{proof}
Using this theorem, we can get rid of the bound on $\alpha$ in \Cref{cor:add_chernoff_bounded}
\AdditiveChernoff*
\begin{proof}
    Case of $\alpha < \mu$ follows from \Cref{cor:add_chernoff_bounded}.

    Otherwise, $\alpha \geq \mu$. Let $Y_i = X_i - \E[X_i]$. Then $\E[Y_i] = 0$. Note that since $X_i \in [0, a]$, we can bound the variances of $Y_1, \ldots, Y_n$:
    \[
        \E[(Y_i - \E[Y_i])^2] \leq \E[X_i^2] \leq a \E[X_i].
    \]
    Therefore,
    \[
        v = \sum_{i \in [n]} \E[(Y_i - \E[Y_i])^2] \leq a \mu.
    \]
    and, in particular, $\frac{a \alpha}{v} \geq 1$.
    Now we can apply \Cref{thm:bennet_ineq} to $Y_1, \ldots, Y_n$:
    \[
        \Pr \left[ X - \mu > \alpha \right] \leq \Pr\left[\sum_{i \in [n]} Y_i \geq \alpha \right] \leq \exp \left(- \frac{\alpha}{2 a} \ln \left( 1 + \frac{a \alpha}{v} \right) \right) \leq \exp \left(- \frac{\alpha}{3 a} \right) \leq \left(- \frac{\eps \alpha}{3 a} \right) 
    \]
    where the second to last inequality follows from $\ln(2)/2 > 1/3$, $1 + \frac{a \alpha}{v} \geq 2$. Finally,
    \[
        \Pr \left[ |X - \mu| > \epsilon \mu + \alpha \right] \leq \Pr \left[ |X - \mu| > \alpha \right] = \Pr \left[ X - \mu > \alpha \right]
    \]
    because $\Pr[\mu - X > \alpha] = 0$ when $\alpha \geq \mu$ since $X \geq 0$.
\end{proof}

\end{document}